\documentclass{article}%
\usepackage{amsmath}
\usepackage{amsfonts}
\usepackage{amssymb}
\usepackage{graphicx}%
\providecommand{\U}[1]{\protect\rule{.1in}{.1in}}
\newtheorem{theorem}{Theorem}

\newtheorem{conclusion}[theorem]{Conclusion}

\newtheorem{corollary}[theorem]{Corollary}

\newtheorem{definition}[theorem]{Definition}

\newtheorem{proposition}[theorem]{Proposition}
\newtheorem{remark}[theorem]{Remark}

\newenvironment{proof}[1][Proof]{\noindent\textbf{#1.} }{\ \rule{0.5em}{0.5em}}
\begin{document}

\author{Brian C. Hall ~and Benjamin D. Lewis\\University of Notre Dame\\Department of Mathematics\\Notre Dame, IN 46556 U.S.A.\\bhall@nd.edu (Hall)\\bdlcalvin@gmail.com (Lewis)}
\title{A unitary \textquotedblleft quantization commutes with
reduction\textquotedblright\ map for the adjoint action of a compact Lie group}
\maketitle

\begin{abstract}
Let $K$ be a simply connected compact Lie group and $T^{\ast}(K)$ its
cotangent bundle. We consider the problem of \textquotedblleft quantization
commutes with reduction\textquotedblright\ for the adjoint action of $K$ on
$T^{\ast}(K).$ We quantize both $T^{\ast}(K)$ and the reduced phase space
using geometric quantization with half-forms. We then construct a
geometrically natural map from the space of invariant elements in the
quantization of $T^{\ast}(K)$ to the quantization of the reduced phase space.
We show that this map is a constant multiple of a unitary map.

\end{abstract}
\tableofcontents

\section{Introduction}

\subsection{Quantization and reduction}

Many classical systems of interest in physics arise as the reduction of some
larger system by the action of a group. Of particular importance are gauge
field theories, in which reduction implements gauge symmetry. In quantizing
such systems, one could plausibly attempt to perform the quantization either
before or after reduction. It is then of interest to compare the results of
these two procedures.

In this paper, we consider the holomorphic approach to quantization, with the
Segal--Bargmann space over $T^{\ast}(\mathbb{R}^{n})\cong\mathbb{C}^{n}$
serving as a prototypical example. This approach to quantization allows one,
for example, to construct a family of coherent states and often facilitates
various aspects of the semiclassical limit. Specifically, we follow the
approach of geometric quantization using a complex (i.e., K\"{a}hler)
polarization. See \cite{holomethods} for background on Segal--Bargmann spaces
and \cite{Woodhouse} or Chapters 22 and 23 of \cite{quantumBook} for
background on geometric quantization.

In the holomorphic approach to quantization, the first major result comparing
quantization before reduction to quantization after reduction is the 1982
paper \cite{GStern} of Guillemin and Sternberg. They work in the setting of
compact K\"{a}hler manifolds, using geometric quantization without half-forms.
Under certain regularity assumptions, they establish a geometrically natural
invertible linear map between the \textquotedblleft first quantize then
reduce\textquotedblright\ space and the \textquotedblleft first reduce and
then quantize\textquotedblright\ space. There have been numerous extensions of
this work, many of which work with a notion of quantization based on the index
of a certain operator. The reader is referred to the survey article of Sjamaar
\cite{Sjamaar} for the state of the art in this area as of 1996.

\subsection{Unitarity in \textquotedblleft quantization commutes with
reduction\textquotedblright}

An issue not addressed in \cite{GStern} is the matter of unitarity of the map.
This issue is important because for physical applications, it is essential to
have not just a vector space but a Hilbert space: The inner product matters!
In \cite{HallKirwin}, the first author and Kirwin analyzed the
Guillemin--Sternberg map and found that in general it is \textit{not even
asymptotically unitary} as Planck's constant tends to zero. The paper
\cite{HallKirwin} then introduces a modified Guillemin--Sternberg-type map in
the presence of half-forms and shows that this map is at least asymptotically
unitary as Planck's constant tends to zero. Much of the analysis in
\cite{HallKirwin} can be applied even if the K\"{a}hler manifold in question
is not compact.

A natural question remains whether there are interesting examples in which a
map of Guillemin--Sternberg type is actually unitary (or unitary up to a
constant), not just asymptotically unitary. Work of Kirwin \cite{Kirwin}
represents a step in the direction of answering this question, by analyzing
the higher-order asymptotics of the map introduced in \cite{HallKirwin}.
Nevertheless, there is no known general criterion for obtaining exact unitarity.

The first example of unitarity we are aware of is in an infinite-dimensional
setting, that of the quantization of $(1+1)$-dimensional Yang--Mills theory on
a space-time cylinder. Fix a connected compact Lie group $K,$ called the
structure group, with Lie algebra $\mathfrak{k}.$ The unreduced configuration
space for the problem is then the space $\mathcal{A}$ of connections (that is,
$\mathfrak{k}$-valued 1-forms) on the spatial circle, and the unreduced phase
space is the cotangent bundle $T^{\ast}(\mathcal{A}).$ We consider at first
the \textbf{based gauge group} $\mathcal{G}_{0}$ consisting of maps from the
circle into $K$ that equal the identity at one fixed point. The symplectically
reduced phase space $T^{\ast}(\mathcal{A})/\!\!/\mathcal{G}_{0}$ is then
naturally identified with the \textit{finite-dimensional} symplectic manifold
$T^{\ast}(K).$ Furthermore, if we choose a natural complex structure on
$T^{\ast}(\mathcal{A}),$ the quotient inherits a complex structure, given by
identifying $T^{\ast}(K)$ with $K_{\mathbb{C}},$ the complexification of $K.$
(This identification is described in Section \ref{prelim.sec}.)

The quantization-versus-reduction problem for holomorphic quantization of
$T^{\ast}(\mathcal{A})$ has been considered using two different methods.
First, Landsman and Wren \cite{LW} and Wren \cite{Wr2} use the
\textquotedblleft generalized Rieffel induction\textquotedblright\ method of
Landsman \cite{La} to \textquotedblleft project\textquotedblright\ the
coherent states for $T^{\ast}(\mathcal{A})$ into the (nonexistent)
gauge-invariant subspace. The article \cite{Wr2} shows that the projected
coherent states are precisely the coherent states obtained in
\cite{Hall94,geoquant} from the holomorphic quantization of $T^{\ast}(K)\cong
K_{\mathbb{C}}.$

Meanwhile, Driver--Hall \cite{DHYM} approach the problem using a Gaussian
measure of large variance to approximate the nonexistent Lebesgue measure on
$\mathcal{A}.$ (See also \cite{RevMathPhys}.) The gauge invariant subspace is
then identified as an $L^{2}$ space of holomorphic functions on $K_{\mathbb{C}%
}$ that converges in the large variance limit to the one obtained in
\cite{geoquant} from quantizing $T^{\ast}(K)\cong K_{\mathbb{C}}.$ Both
\cite{Wr2} and \cite{DHYM} identify the space obtained by first quantizing and
then reducing as being \textit{the same space with the same inner product} as
the one obtained by first reducing and then quantizing.

The goal of the present paper is to provide a finite-dimensional example of a
unitary \textquotedblleft quantization commutes with
reduction\textquotedblright\ map.

\subsection{The adjoint action of $K$ on $T^{\ast}(K)$}

We now let $K$ be a connected compact Lie group, which we assume for
simplicity to be simply connected. The cotangent bundle $T^{\ast}(K)$ has a
natural K\"{a}hler structure obtained by identification of $T^{\ast}(K)$ with
the complexified group $K_{\mathbb{C}}.$ (See Section \ref{prelim.sec}.)
Geometric quantization of $T^{\ast}(K)\cong K_{\mathbb{C}},$ using a
K\"{a}hler polarization and half-forms, is carried out in \cite{geoquant}. The
quantum Hilbert space turns out to be naturally isomorphic to the generalized
Segal--Bargmann space over $K_{\mathbb{C}}$ introduced in \cite{Hall94}.
Furthermore, the BKS pairing map between the K\"{a}hler-polarized and
vertically polarized spaces turns out to coincide with the generalized
Segal--Bargmann transform of \cite{Hall94}; in particular, the pairing map is
unitary in this case. Related aspects of the geometric quantization of
$T^{\ast}(K)\cong K_{\mathbb{C}}$ have been studied by Florentino, Matias,
Mour\~{a}o, and Nunes \cite{FMMN1,FMMN2}, by Lempert and Sz\H{o}ke
\cite{LS2,LS3}, and by Sz\H{o}ke \cite{Sz2}. All of these authors study
\textit{families} of complex structures on $T^{\ast}(K)$ and parallel
transport in the resulting bundle of quantum Hilbert spaces.

Meanwhile, Florentino, Mour\~{a}o, and Nunes \cite[Theorems 2.2 and 2.3]{FMN}
give a remarkable explicit formula for norm of a holomorphic class function on
$K_{\mathbb{C}}$ (in the Hilbert space obtained by quantizing $T^{\ast
}(K)\cong K_{\mathbb{C}}$) as an integral over a complex maximal torus. The
formula involves the Weyl denominator function and can be viewed as a complex
version of the Weyl integral formula. The results of \cite{FMN} are a vital
ingredient in the present paper. Indeed, our main result can be summarized by
saying that the just-cited result of \cite{FMN} can be interpreted as a
unitary \textquotedblleft quantization commutes with reduction
map.\textquotedblright\ 

Now $K$ acts on itself by the adjoint action (i.e., by conjugation), and this
action lifts to a symplectic action of $K$ on the cotangent bundle $T^{\ast
}(K).$ There is then an equivariant momentum map%
\[
\phi:T^{\ast}(K)\rightarrow\mathfrak{k}^{\ast}.
\]
Although the action of $K$ on $\phi^{-1}(0)$ is not free (even generically),
we may attempt to construct the symplectic quotient
\[
T^{\ast}(K)/\!\!/\mathrm{Ad}_{K}:=\phi^{-1}(0)/\mathrm{Ad}_{K}.
\]
The quotient will not be a manifold, but will have singularities that must be
dealt with in the analysis.

The reduction of $T^{\ast}(K)$ is a natural problem for various reasons.
First, from the point of view of the Yang--Mills example described in the
previous subsection, we have said that $T^{\ast}(K)$ is the symplectic
quotient of $T^{\ast}(\mathcal{A})$ by the \textit{based} gauge group
$\mathcal{G}_{0}.$ Let $\mathcal{G}$ denote the full gauge group, consisting
of all maps of $S^{1}$ into $K.$ The adjoint action of $K$ on $T^{\ast}(K)$ is
then the \textquotedblleft residual\textquotedblright\ action of the full
gauge group $\mathcal{G}$ on $T^{\ast}(\mathcal{A})/\!\!/\mathcal{G}_{0}.$
Second, the quotient $T^{\ast}(K)/\!\!/\mathrm{Ad}_{K}$ is a geometrically
interesting example of a singular symplectic quotient. Quantization of this
quotient has been studied by several researchers, including Wren \cite{Wr1},
Huebschmann \cite{Hueb}, Huebschmann, Rudolph and Schmidt \cite{HRS}, and
Boeijink, Landsman, and van Suijlekom \cite{BLV}. Third, the quotient
$T^{\ast}(K)/\!\!/\mathrm{Ad}_{K}$ is a prototype---the case of a single
plaquette---for the study of lattice gauge systems. See \cite{Fu} for a study
of more general cases.

The paper \cite{BLV} of Boeijink, Landsman, and van Suijlekom, in particular,
considers the quotient from the point of view of the quantization versus
reduction problem. The authors consider \textquotedblleft Dolbeault--Dirac
quantization\textquotedblright\ of both $T^{\ast}(K)$ and of (the regular
points in) $T^{\ast}(K)/\!\!/\mathrm{Ad}_{K}.$ The Dolbeault--Dirac
quantization ultimately turns out to give the same result as geometric
quantization of these spaces \textit{without} half-forms \cite[Theorem
3.14]{BLV}. (The authors also consider \textquotedblleft spin
quantization\textquotedblright\ of $T^{\ast}(K),$ which ultimately gives the
same result \cite[Theorem 3.15]{BLV} as geometric quantization \textit{with}
half-forms, but their results on \textquotedblleft quantization commutes with
reduction\textquotedblright\ are for the Dolbeault--Dirac quantization.) The
authors determine (1) the invariant subspace of the quantization of $T^{\ast
}(K),$ and (2) the quantization of $T^{\ast}(K)/\!\!/\mathrm{Ad}_{K}.$ They
then show that both of these spaces can be identified unitarily with the space
of Weyl-invariant elements in the Hilbert space $L^{2}(T)$ \cite[Theorem
4.18]{BLV}. Although this result constitutes a form of \textquotedblleft
quantization commutes with reduction,\textquotedblright\ the isomorphism
between these spaces is constructed in an indirect way, making use of a very
general Segal--Bargmann-type isomorphism from \cite[Section 10]{Hall94}.
Indeed, the authors say, \textquotedblleft the quantization commutes with
reduction theorem would get more body if there were a way to identify
quantization after reduction with reduction after quantization differently
from mere unitary isomorphism of Hilbert spaces\textquotedblright\ \cite[p.
31]{BLV}.

By contrast, we will consider quantization of $T^{\ast}(K)$ and $T^{\ast
}(K)/\!\!/\mathrm{Ad}_{K}$ with half-forms. We will consider a natural map of
Guillemin--Sternberg type between the space of invariant elements in the
quantization of $T^{\ast}(K)$ (\textquotedblleft first quantize and then
reduce\textquotedblright) and the quantization of $T^{\ast}%
(K)/\!\!/\mathrm{Ad}_{K}$ (\textquotedblleft first reduce and then
quantize\textquotedblright). This map will be similar to the one constructed
in \cite{HallKirwin} and will include a mechanism for converting half-forms of
one degree to half-forms of a smaller degree. Our main result will be that
this geometrically natural map is a constant multiple of a unitary map.

\subsection{The main results\label{main.sec}}

We consider geometric quantization of $T^{\ast}(K)$ with half-forms, using a
K\"{a}hler polarization obtained by identifying $T^{\ast}(K)$ with
$K_{\mathbb{C}}.$ We denote the resulting Hilbert space by $\mathrm{Quant}%
(K_{\mathbb{C}})$. (Whenever we write $K_{\mathbb{C}}$, we always mean
\textquotedblleft$K_{\mathbb{C}}$ as identified with $T^{\ast}(K)$%
.\textquotedblright) The adjoint action of $K$ on $T^{\ast}(K)\cong
K_{\mathbb{C}}$ then induces an unitary \textquotedblleft
adjoint\textquotedblright\ action of $K$ on $\mathrm{Quant}(K_{\mathbb{C}}).$
We let $\mathrm{Quant}(K_{\mathbb{C}})^{\mathrm{Ad}_{K}}$ denote the space of
elements in $\mathrm{Quant}(K_{\mathbb{C}})$ that are fixed by this action. We
think of $\mathrm{Quant}(K_{\mathbb{C}})^{\mathrm{Ad}_{K}}$ as being the
reduced quantum Hilbert space, that is, the space obtained by \textit{first}
quantizing and \textit{then} reducing.

Now let $T\subset K$ be a fixed maximal torus, let $T^{\ast}(T)$ be its
cotangent bundle, and let $T_{\mathbb{C}}\subset K_{\mathbb{C}}$ be the
complexification of $T.$ Since $T$ is also a compact Lie group, we may
similarly identify $T^{\ast}(T)$ with $T_{\mathbb{C}}$ and perform geometric
quantization with half-forms. We let $\mathrm{Quant}(T_{\mathbb{C}})$ denote
the quantization of $T^{\ast}(T)\cong T_{\mathbb{C}}.$ If $W$ denotes the Weyl
group, we then identify a \textquotedblleft Weyl-alternating\textquotedblright%
\ subspace of $\mathrm{Quant}(T_{\mathbb{C}}),$ which we denote as
$\mathrm{Quant}(T_{\mathbb{C}})^{W_{-}}.$

The adjoint action of $K$ on $T^{\ast}(K)$ admits an equivariant momentum map
$\phi.$ We consider, finally, the reduced phase space%
\[
T^{\ast}(K)/\!\!/\mathrm{Ad}_{K}:=\phi^{-1}(0)/\mathrm{Ad}_{K},
\]
which we also write as $K_{\mathbb{C}}/\!\!/\mathrm{Ad}_{K}$. Since the
reduced phase space is not a manifold, we will quantize it by quantizing only
the set of regular points---what is called the \textquotedblleft principal
stratum\textquotedblright\ in \cite{HRS} and \cite{BLV}---which is an open
dense subset. (An argument for the reasonableness of this procedure is given
in Section \ref{reducedQuant.sec}.) We denote the quantization of the reduced
phase space by $\mathrm{Quant}(K_{\mathbb{C}}/\!\!/\mathrm{Ad}_{K}).$

Now, in Section \ref{tkQuant.sec}, we will see that $\mathrm{Quant}%
(K_{\mathbb{C}})$ and $\mathrm{Quant}(T_{\mathbb{C}})$ can be identified as
$L^{2}$ spaces of holomorphic functions with respect to certain measures
$\gamma_{\hbar}$ and $\gamma_{\hbar}^{\prime},$ respectively. Then in Section
\ref{invariantSubspace.sec}, we will see that $\mathrm{Quant}(K_{\mathbb{C}%
})^{\mathrm{Ad}_{K}}$ corresponds to the space of holomorphic class functions
in $L^{2}(K_{\mathbb{C}},\gamma_{\hbar}).$ Let $\sigma:T\rightarrow\mathbb{R}$
be the Weyl denominator function and let $\sigma_{\mathbb{C}}:T_{\mathbb{C}%
}\rightarrow\mathbb{C}$ be the analytic continuation of $\sigma.$ (Since $K$
is assumed simply connected, $\sigma$ is a single-valued function on $T.$) We
now record a crucial result of Florentino, Mour\~{a}o, and Nunes \cite{FMN}.

\begin{theorem}
\label{fmnIntro.thm}There is a constant $c_{\hbar}$ such that for every
holomorphic class function $F$ on $K_{\mathbb{C}}$, we have%
\[
\int_{K_{\mathbb{C}}}\left\vert F(g)\right\vert ^{2}~d\gamma_{\hbar
}(g)=c_{\hbar}\int_{T_{\mathbb{C}}}\left\vert \sigma_{\mathbb{C}%
}(z)F(z)\right\vert ^{2}~d\gamma_{\hbar}^{\prime}(z).
\]

\end{theorem}

This result is Theorem 2.2 in \cite{FMN}. The goal of the present paper is
this: To interpret Theorem \ref{fmnIntro.thm} as giving a unitary (up to a
constant) \textquotedblleft quantization commutes with
reduction\textquotedblright\ map between $\mathrm{Quant}(K_{\mathbb{C}%
})^{\mathrm{Ad}_{K}}$ and $\mathrm{Quant}(K_{\mathbb{C}}/\!\!/\mathrm{Ad}%
_{K}).$ To accomplish this goal, we must perform two tasks. First, we must
show that we can quantize the reduced phase space $K_{\mathbb{C}%
}/\!\!/\mathrm{Ad}_{K}$ in such a way that $\mathrm{Quant}(K_{\mathbb{C}%
}/\!\!/\mathrm{Ad}_{K})$ may be identified with $\mathrm{Quant}(T_{\mathbb{C}%
})^{W_{-}}.$ Second, we must show that the map%
\begin{equation}
F\mapsto(\sigma_{\mathbb{C}})(\left.  F\right\vert _{T_{\mathbb{C}}})
\label{Fmapsto}%
\end{equation}
implicit in Theorem \ref{fmnIntro.thm} can be computed by means of a natural
Guillemin--Sternberg-type map with half-forms.

In the second task, the key issue is to account for the appearance of the
analytically continued Weyl denominator $\sigma_{\mathbb{C}}$ in
(\ref{Fmapsto}). We will argue that this factor is not something one simply
needs to insert by hand in order to obtain the nice result in Theorem
\ref{fmnIntro.thm}. Rather, we will show that this factor arises naturally in
the process of converting half-forms on $K_{\mathbb{C}}$ to half-forms (of a
different degree!) on $T_{\mathbb{C}}.$ Specifically, we will use a
procedure---similar to that in \cite{HallKirwin}---of contracting half-forms
with the vector fields representing the infinitesimal adjoint action of $K$ on
$K_{\mathbb{C}}.$ The analytically continued Weyl denominator $\sigma
_{\mathbb{C}}$ will arise naturally in this process. (See Sections
\ref{relatingCanonical.sec} and \ref{relatingHalf.sec}.)

We describe our results in schematic form.

\begin{theorem}
\label{main1.thm}It is possible to quantize the reduced phase space in such
way that the elements of the quantum Hilbert space may be identified with the
Weyl-alternating elements of the quantization of $T_{\mathbb{C}}$:%
\[
\mathrm{Quant}(K_{\mathbb{C}}/\!\!/\mathrm{Ad}_{K})\cong\mathrm{Quant}%
(T_{\mathbb{C}})^{W_{-}}.
\]

\end{theorem}

In proving this result, we will exploit the freedom that is present in
geometric quantization to choose the prequantum line bundle.

\begin{theorem}
\label{main2.thm}There is a geometrically natural \textquotedblleft
quantization commutes with reduction\textquotedblright\ map
\[
B:\mathrm{Quant}(K_{\mathbb{C}})^{\mathrm{Ad}_{K}}\rightarrow\mathrm{Quant}%
(K_{\mathbb{C}}/\!\!/\mathrm{Ad}_{K})\cong\mathrm{Quant}(T_{\mathbb{C}%
})^{W_{-}}%
\]
that corresponds, after suitable identifications, to the map $F\mapsto
(\sigma_{\mathbb{C}})(\left.  F\right\vert _{T_{\mathbb{C}}})$ in
(\ref{Fmapsto}). Thus, by Theorem \ref{fmnIntro.thm} the map $B$ is a constant
multiple of a unitary map.
\end{theorem}

Although the map $B$ is a similar to the map $B_{k}$ in \cite{HallKirwin},
modifications are needed in the present setting. First, because the adjoint
action of $K$ on $\phi^{-1}(0)$ is not even generically free, we must contract
at each point with a subset of the vector fields representing the
infinitesimal adjoint action. Second, the contraction process requires a
choice of orientation and it turns out to be impossible to choose the
orientation consistently over all of $\phi^{-1}(0).$ It is therefore necessary
to choose the prequantum line bundle in the quantization of the reduced phase
space carefully in order to ensure that the \textquotedblleft quantization
commutes with reduction\textquotedblright\ map is globally defined.

\section{Preliminaries\label{prelim.sec}}

Let $K$ be a connected compact Lie group of dimension $n$, assumed for
simplicity to be simply connected. This assumption ensures that the Weyl
denominator function (Section \ref{maximalTorus.sec}) is single valued and
that the centralizer of each regular semisimple element in the complexified
group is a complex maximal torus (Section \ref{regularPoints.sec}). We fix an
inner product $\left\langle \cdot,\cdot\right\rangle $ on the Lie algebra
$\mathfrak{k}$ of $K$ that is invariant under the adjoint action of $K.$ There
is then a unique bi-invariant Riemannian metric on $K$ whose value at the
identity is $\left\langle \cdot,\cdot\right\rangle .$ We let $n$ denote the
dimension of $K.$

\subsection{The complex structure on $T^{\ast}(K)$}

We let $K_{\mathbb{C}}$ be the \textbf{complexification} of $K,$ which may be
described as the unique simply connected Lie group whose Lie algebra is
$\mathfrak{k}_{\mathbb{C}}:=\mathfrak{k}\oplus i\mathfrak{k}.$ The inclusion
of $\mathfrak{k}$ into $\mathfrak{k}_{\mathbb{C}}$ induces a homomorphism of
$K$ into $K_{\mathbb{C}},$ which is well known to be injective. (See
\cite[Section 3]{Hall94}.) We will identify $K$ with its image inside
$K_{\mathbb{C}},$ which is a compact and therefore closed subgroup of
$K_{\mathbb{C}}.$ As an example, we may take $K$ to be the special unitary
group $SU(N)$ and $K_{\mathbb{C}}$ to be the special linear group
$SL(N;\mathbb{C}).$

As our initial phase space (before reduction) we take the cotangent bundle
$T^{\ast}(K),$ with the canonical 2-form $\omega.$ Results of Lempert and
Sz\H{o}ke \cite{LS,Szoke1} and of Guillemin and Stenzel \cite{GStenz1,GStenz2}
show that there is a natural globally defined \textquotedblleft adapted
complex structure\textquotedblright\ on $T^{\ast}(K)$ determined by the choice
of bi-invariant metric on $K.$ This complex structure can be described
explicitly as follows. First, use left translation to identify $T^{\ast}(K)$
with $K\times\mathfrak{k}^{\ast}.$ Then use the inner product on
$\mathfrak{k}$ to identify $K\times\mathfrak{k}^{\ast}$ with $K\times
\mathfrak{k}.$ Finally, use the diffeomorphism $\Psi:K\times\mathfrak{k}%
\rightarrow K_{\mathbb{C}}$ given by the polar decomposition for
$K_{\mathbb{C}}$:%
\begin{equation}
\Psi(x,\xi)=xe^{-i\xi},\quad x\in K,~\xi\in\mathfrak{k}.\label{PhiDef}%
\end{equation}
We then use $\Psi$ to pull back the complex structure on $K_{\mathbb{C}}$ to
$T^{\ast}(K).$ The resulting complex structure on $T^{\ast}(K)$ fits together
with the symplectic structure to make $T^{\ast}(K)$ into a K\"{a}hler
manifold. This claim is a consequence of general results in the theory of
adapted complex structures (e.g., the theorem on p. 568 of \cite{GStenz1}),
but can be verified directly by the calculations in the first appendix of
\cite{geoquant}.

In this paper, we follow the sign conventions in \cite{quantumBook}. With
these conventions, the minus sign in the exponent on the right-hand side of
(\ref{PhiDef}) is necessary in order to achieve the positivity condition in
the definition of a K\"{a}hler manifold. Actually, a similar minus sign is
needed even in the case of $T^{\ast}(\mathbb{R}).$ If the canonical 2-form is
defined as $\omega=dp\wedge dx$ (as in \cite{quantumBook}), then the complex
structure must be defined as $z=x-ip$ rather than $x+ip$ in order to for
$\omega(X,JX)$ to be non-negative.

\subsection{Maximal tori\label{maximalTorus.sec}}

We fix throughout the paper a maximal torus $T$ of $K$ and we denote its Lie
algebra by $\mathfrak{t}.$ We let $r$ denote the dimension of $T.$ We let
$T_{\mathbb{C}}$ denote the connected subgroup of $K_{\mathbb{C}}$ whose Lie
algebra is $\mathfrak{t}_{\mathbb{C}}:=\mathfrak{t}\oplus i\mathfrak{t}.$ If
$T$ is isomorphic to $(S^{1})^{r},$ it follows from the polar decomposition
for $K_{\mathbb{C}}$ that $T_{\mathbb{C}}$ is isomorphic to $(\mathbb{C}%
^{\ast})^{r},$ so that $T_{\mathbb{C}}$ is the complexification of $T$ in the
sense of Section 3 of \cite{Hall94}.

The bi-invariant metric on $K$ restricts to an invariant metric on $T.$ We may
then regard the cotangent bundle $T^{\ast}(T)$ as a submanifold of $T^{\ast
}(K)$ using the metrics:%
\begin{equation}
T^{\ast}(T)\cong T(T)\subset T(K)\cong T^{\ast}(K). \label{TTInTK}%
\end{equation}
We may also identify the cotangent bundle $T^{\ast}(T)$ with $T_{\mathbb{C}}$
identifying $T^{\ast}(T)$ with $T\times\mathfrak{t}^{\ast}$ and then with
$T\times\mathfrak{t}$ and then applying the map $\Psi^{\prime}:T\times
\mathfrak{t}\rightarrow T_{\mathbb{C}}$ given by the same formula as in
(\ref{PhiDef}):%
\begin{equation}
\Psi^{\prime}(t,H)=te^{-iH},\quad t\in T,~H\in\mathfrak{t}. \label{PhiT}%
\end{equation}

We say that a Lie subgroup $S$ of $K_{\mathbb{C}}$ is a complex torus if it is
isomorphic as a complex Lie group to a direct product of copies of
$\mathbb{C}^{\ast}.$ We say that $S$ is a \textbf{complex maximal torus} if it
is a complex torus that is not properly contained in another complex torus.
The group $T_{\mathbb{C}}$ is a complex maximal torus. and all complex maximal
tori are conjugate. (See Corollary A to Theorem 21.3 in \cite{HumLAG}.)

We let $R\subset\mathfrak{t}$ denote the \textbf{root system} associated to
the pair $(K,T).$ (Specifically, $R$ is the set of \textquotedblleft real
roots,\textquotedblright\ in the sense of \cite[Definition 11.34]{LieBook}.)
We fix once and for all a set $R^{+}$ of \textbf{positive roots}. We also let
$W:=N(T)/T$ denote the \textbf{Weyl group}. By Theorem 11.36 in \cite{LieBook}%
, $W$ may be identified with the subgroup of the orthogonal group
$O(\mathfrak{t})$ generated by the reflections about the hyperplanes
perpendicular to the roots. The adjoint action of $W$ on $T$ extends to an
action $T_{\mathbb{C}}.$

We let $\sigma:T\rightarrow\mathbb{R}$ denote the \textbf{Weyl denominator},
given by%
\begin{equation}
\sigma(e^{H})=(2i)^{m}\prod_{\alpha\in R^{+}}\sin\left(  \frac{\left\langle
\alpha,H\right\rangle }{2}\right)  ,\quad H\in\mathfrak{t}, \label{weylDenom1}%
\end{equation}
where $m$ is the number of positive roots, or by the alternative expression,%
\begin{equation}
\sigma(e^{H})=\sum_{w\in W}\mathrm{sign}(w)e^{i\left\langle w\cdot
\delta,H\right\rangle }, \label{weylDenom2}%
\end{equation}
where $\delta$ is half the sum of the positive roots. (See \cite[Lemma
10.28]{LieBook} for the equality of these two expressions.) Since $K$ is
simply connected, $\delta$ is an analytically integral element \cite[Corollary
13.21]{LieBook} and $\sigma$ is therefore a single-valued function on $T.$ We
also let $\sigma_{\mathbb{C}}:T_{\mathbb{C}}\rightarrow\mathbb{C}$ be the
analytic continuation of the Weyl denominator, which is given by either of the
expressions (\ref{weylDenom1}) or (\ref{weylDenom2}), but with $H$ now
belonging to $\mathfrak{t}_{\mathbb{C}}.$

We say that a function $f$ on $T$ or $T_{\mathbb{C}}$ is \textbf{Weyl
alternating} if%
\[
f(w\cdot z)=\mathrm{sign}(w)f(z)
\]
for all $z$ in $T$ or $T_{\mathbb{C}}.$ Using either of the expressions for
the Weyl denominator, one easily shows that $\sigma$ and $\sigma_{\mathbb{C}}$
are Weyl alternating.

\subsection{The momentum map\label{momentumMap.sec}}

We refer to Section 4.2 in \cite{AM} for general information about momentum
maps. We consider the adjoint action of $K$ on itself, given by%
\[
x\cdot y=xyx^{-1},
\]
and also the induced adjoint action of $K$ on $T^{\ast}(K).$ Since the action
of $K$ on $T^{\ast}(K)$ is induced from an action on the base, there is an
equivariant \textbf{momentum map}
\[
\phi:T^{\ast}(K)\rightarrow\mathfrak{k}^{\ast}%
\]
that is linear on each fiber of $T^{\ast}(K).$ (See Corollary 4.2.11 in
\cite{AM}.) To describe $\phi,$ let us introduce the following notation. For
each $\eta\in\mathfrak{k},$ we define the $\eta$\textbf{-component} of $\phi$
to be the function $\phi_{\eta}:T^{\ast}(K)\rightarrow\mathbb{R}$ be given by%
\[
\phi_{\eta}(x,\xi)=\phi(x,\xi)(\eta).
\]
(These functions have the property that the Hamiltonian flow generated by
$\phi_{\eta}$ is just the action adjoint action of the one-parameter subgroup
of $K$ generated by $\eta.$) Then $\phi$ is determined by the following
formula%
\[
\phi_{\eta}(x,\xi)=\xi(Y^{\eta}(x)),
\]
where $Y^{\eta}$ is the vector field representing the infinitesimal adjoint
action of $\eta$ on $K.$

Let us identify $T^{\ast}(K)$ with $K\times\mathfrak{k}$ using left
translation and the inner product on $\mathfrak{k}.$ Then we may easily
compute that $Y^{\eta}(x)=\mathrm{Ad}_{x^{-1}}(\eta)-\eta,$ so that%
\[
\phi_{\eta}(x,\xi)=\left\langle \xi,\mathrm{Ad}_{x^{-1}}(\eta)-\eta
\right\rangle =\left\langle \mathrm{Ad}_{x}(\xi)-\xi,\eta\right\rangle .
\]
Thus, the momentum map, viewed as a map of $K\times\mathfrak{k}$ into
$\mathfrak{k}^{\ast}\cong\mathfrak{k}$ is given explicitly as%
\begin{equation}
\phi(x,\xi)=\mathrm{Ad}_{x}(\xi)-\xi. \label{momentum}%
\end{equation}

\section{Reduction of the quantum Hilbert space}

In this section we consider the Hilbert space obtained by \textit{first}
quantizing the phase space $T^{\ast}(K)$ and \textit{then} reducing by the
adjoint action of $K,$ which we write as $\mathrm{Quant}(T^{\ast
}(K))^{\mathrm{Ad}_{K}}.$ We will identify $\mathrm{Quant}(T^{\ast}(K))$ as an
$L^{2}$ space of holomorphic functions on $K_{\mathbb{C}}$ and $\mathrm{Quant}%
(T^{\ast}(K))^{\mathrm{Ad}_{K}}$ as the corresponding $L^{2}$ space of
holomorphic class functions on $K_{\mathbb{C}}$.

\subsection{Quantization of the cotangent bundle\label{tkQuant.sec}}

We briefly explain some of the results in \cite{geoquant}, in which the phase
space $T^{\ast}(K)\cong K_{\mathbb{C}}$ is quantized using geometric
quantization with half-forms. We follow the sign conventions in the book
\cite[Chapters 22 and 23]{quantumBook}, which differ from those in
\cite{geoquant}. We let $\omega$ denote the canonical 2-form on $T^{\ast}(K),$
given in local coordinates as $\omega=\sum dp_{j}\wedge dx_{j}.$ We let
$\theta$ be the canonical 1-form on $T^{\ast}(K),$ satisfying $d\theta
=\omega.$ We let $L=T^{\ast}(K)\times\mathbb{C}$ be the trivial line bundle
over $T^{\ast}(K),$ so that sections of $L$ are identified with complex-valued
functions on $T^{\ast}(K).$ We use the trivial Hermitian structure on $L,$ so
that the magnitude of a section is just the absolute value of the
corresponding function. We define a connection $\nabla$ on $L$ by setting%
\begin{equation}
\nabla_{X}f=Xf-\frac{i}{\hbar}\theta(X)f \label{covariantDeriv}%
\end{equation}
for each smooth section (i.e., function) $f$ and each vector field $X.$

We say that a smooth section $f$ of $L$ is a \textbf{holomorphic section} if%
\[
\nabla_{X}f=0
\]
for all vector fields $X$ of type $(0,1)$ on $T^{\ast}(K)\cong K_{\mathbb{C}%
}.$ Although we identify sections with functions, the \textit{holomorphic}
sections do not correspond to holomorphic functions. Rather, the function
$\kappa(x,\xi)=\left\vert \xi\right\vert ^{2}$ is a K\"{a}hler potential for
$T^{\ast}(K)\cong K_{\mathbb{C}}.$ This claim follows from a general result
\cite[p. 568]{GStenz1} about adapted complex structures, and is verified by
direct computation in the present case in the first appendix to
\cite{geoquant}. It then follows easily that the holomorphic sections are
precisely those of the form%
\[
f=Fe^{-\left\vert \xi\right\vert ^{2}/(2\hbar)},
\]
where $F$ is a holomorphic function on $K_{\mathbb{C}}.$ Here the expression
\textquotedblleft$\xi$\textquotedblright\ is defined as a function on
$K_{\mathbb{C}}$ by means of the diffeomorphism $\Psi$ in (\ref{PhiDef}).

The \textbf{canonical bundle }$\mathcal{K}$ for $T^{\ast}(K)\cong
K_{\mathbb{C}}$ is the holomorphic line bundle whose holomorphic sections are
holomorphic $n$-forms, where $n$ is the complex dimension of $K_{\mathbb{C}}.$
The canonical bundle is holomorphically trivial, and we will choose a
nowhere-vanishing, left-$K_{\mathbb{C}}$-invariant holomorphic $n$-form
$\beta.$ (The form $\beta$ is unique up to a constant.) We then take a trivial
square root $\mathcal{K}_{1/2}$ to the canonical bundle, with a trivializing
section $\sqrt{\beta}$ satisfying%
\[
\sqrt{\beta}\otimes\sqrt{\beta}=\beta.
\]

We define a Hermitian structure on $\mathcal{K}_{1/2}$ by setting%
\begin{equation}
\left\vert \sqrt{\beta}\right\vert ^{2}=\left[  \frac{\beta\wedge\bar{\beta}%
}{b\varepsilon}\right]  ^{1/2}, \label{hermitianK12}%
\end{equation}
where $\varepsilon$ is the Liouville volume form on $T^{\ast}(K)\cong
K_{\mathbb{C}}$:%
\[
\varepsilon:=\frac{\omega^{n}}{n!},
\]
and where $b$ is chosen so that at each point $\beta\wedge\bar{\beta}$ is a
positive multiple of $b\varepsilon.$ We may take, for example,%
\[
b=(2i)^{n}(-1)^{n(n-1)/2}.
\]
The quotient $\beta\wedge\bar{\beta}/(b\varepsilon)$ should be interpreted as
the unique function $j$ such that $\beta\wedge\bar{\beta}=jb\varepsilon.$

The elements of the \textbf{unreduced quantum Hilbert space} $\mathrm{Quant}%
(K_{\mathbb{C}})$ are square integrable holomorphic sections of $L\otimes
\mathcal{K}_{1/2}.$ Each section $\psi$ can be expressed uniquely as%
\[
\psi=Fe^{-\left\vert \xi\right\vert ^{2}/(2\hbar)}\otimes\sqrt{\beta},
\]
where $F$ is a holomorphic function on $K_{\mathbb{C}}.$ The norm of such a
section is computed as%
\[
\left\Vert \psi\right\Vert ^{2}=\int_{K_{\mathbb{C}}}\left\vert
F(g)\right\vert ^{2}e^{-\left\vert \xi\right\vert ^{2}/\hbar}\eta\varepsilon,
\]
where
\begin{equation}
\eta=\left[  \frac{\beta\wedge\bar{\beta}}{b\varepsilon}\right]  ^{1/2}
\label{eta}%
\end{equation}
is the function on the right-hand side of (\ref{hermitianK12}). An explicit
formula for $\eta$ is given in Eq. (2.10) of \cite{geoquant}.

\begin{conclusion}
\label{kcQuant.conclusion}We may quantize the phase space $T^{\ast}(K)\cong
K_{\mathbb{C}}$ in such a way that each element $\psi$ of the $\mathrm{Quant}%
(K_{\mathbb{C}})$ has the form%
\[
\psi=Fe^{-\left\vert \xi\right\vert ^{2}/(2\hbar)}\otimes\sqrt{\beta},
\]
where $F$ is a holomorphic function on $K_{\mathbb{C}}.$ The norm of $\psi$ is
the $L^{2}$ norm of $F$ with respect to the measure%
\begin{equation}
\gamma_{\hbar}:=e^{-\left\vert \xi\right\vert ^{2}/\hbar}\eta\varepsilon,
\label{gamma}%
\end{equation}
where $\varepsilon$ is the Liouville volume measure and $\eta$ is as in
(\ref{eta}).
\end{conclusion}

The preceding result is a straightforward computation, first done in
\cite{geoquant}, using the methods of geometric quantization. What is
remarkable about the result is that the measure $\gamma_{\hbar}$ coincides up
to a constant with a measure on $K_{\mathbb{C}}$ introduced from a very
different point of view in \cite{Hall94}.

\begin{proposition}
\label{geoHeat.prop}For each $\hbar>0,$ there is a constant $c_{\hbar}>0$ such
that the measure $\gamma_{\hbar}$ in (\ref{gamma}) coincides with the
\textquotedblleft$K$-averaged heat kernel measure\textquotedblright%
\ $\nu_{\hbar}(g)~dg$ occurring in \cite[Theorem 2]{Hall94}.
\end{proposition}

The paper \cite{geoquant} also considers the \textquotedblleft BKS pairing
map\textquotedblright\ between the quantization of $T^{\ast}(K)\cong
K_{\mathbb{C}}$ obtained using the K\"{a}hler polarization and the
quantization obtained using the vertical polarization. The result is that the
pairing map coincides up to a constant with the generalized Segal--Bargmann
transform introduced in \cite{Hall94}. In particular, the pairing map is a
constant multiple of a unitary map, something that is certainly not true for a
typical pair of polarizations on a symplectic manifold.

For the purposes of the present paper, the importance of Proposition
\ref{geoHeat.prop} is that it is used in the proof of a critical
result---described in Section \ref{fmnThm.sec}---of Florentino, Mour\~{a}o,
and Nunes.

\subsection{The invariant subspace\label{invariantSubspace.sec}}

For each smooth function $\phi$ on $T^{\ast}(K),$ we let $X_{\phi}$ denote the
associated Hamiltonian vector field, which satisfies $\omega(X_{\phi}%
,\cdot)=d\phi.$ We then define the prequantum operator $Q_{\mathrm{pre}}%
(\phi),$ acting on the space of smooth sections of $L,$ by%
\begin{equation}
Q_{\mathrm{pre}}(\phi)=i\hbar\nabla_{X_{\phi}}+\phi. \label{preQuant}%
\end{equation}
If the Hamiltonian flow generated by $\phi$ preserves the polarization on
$T^{\ast}(K)$---that is, if the Hamiltonian flow is holomorphic on
$K_{\mathbb{C}}$---then we can define a (typically unbounded) quantum operator
$Q(\phi)$ on the quantum Hilbert space by the formula%
\[
Q(\phi)[f\otimes\sqrt{\beta}]=(Q_{\mathrm{pre}}(\phi)f)\otimes\sqrt{\beta
}+i\hbar f\otimes\left[  \frac{1}{2}\frac{\mathcal{L}_{X_{\phi}}(\beta)}%
{\beta}\sqrt{\beta}\right]  .
\]
(There is a typographical error in Definition 23.52 of \cite{quantumBook}; the
sign on the right-hand side should be plus rather than minus.)

We now specialize to the case in which $\phi$ is $\phi_{\eta},$ one of the
components of the momentum map. Under our identification of $T^{\ast}(K)$ with
$K_{\mathbb{C}},$ the adjoint action of $K$ on $T^{\ast}(K)$ corresponds to
the conjugation action of $K$ on $K_{\mathbb{C}}$, which is holomorphic. Thus,
$Q(\phi_{\eta})$ is a well-defined operator on the quantum Hilbert space.

\begin{definition}
\label{invariant.def}We say that an element $\psi$ of the quantization of
$T^{\ast}(K)$ is \textbf{invariant} if%
\[
Q(\phi_{\eta})\psi=0
\]
for all $\eta\in\mathfrak{k}.$ The \textbf{reduced quantum Hilbert space} is
the space of all invariant sections.
\end{definition}

We now compute the space of invariant sections explicitly.

\begin{proposition}
Suppose we write an element $\psi$ of the quantum Hilbert space as%
\[
\psi=Fe^{-\left\vert \xi\right\vert ^{2}/(2\hbar)}\otimes\sqrt{\beta}%
\]
as in Conclusion \ref{kcQuant.conclusion}. Then $\psi$ is invariant in the
sense of Definition \ref{invariant.def} if and only if the holomorphic
function $F$ is a class function on $K_{\mathbb{C}}.$ Thus, the reduced
quantum Hilbert space $\mathrm{Quant}(K_{\mathbb{C}})^{\mathrm{Ad}_{K}}$ is
the space of holomorphic class functions on $K_{\mathbb{C}}$ that are square
integrable with respect to the measure $\gamma_{\hbar}$ in (\ref{gamma}).
\end{proposition}

\begin{proof}
We first make an observation about the operator $Q_{\mathrm{pre}}(\phi)$ in
(\ref{preQuant}). By the definition (\ref{covariantDeriv}) of the covariant
derivative, we have%
\[
Q_{\mathrm{pre}}(\phi)=i\hbar X_{\phi}+\theta(X_{\phi})+\phi,
\]
where $\theta$ satisfies $d\theta=\omega.$ Now, by Cartan's formula for the
Lie derivative
\begin{align*}
\mathcal{L}_{X_{\phi}}\theta &  =d[i_{X_{\phi}}\theta]+i_{X_{\phi}}(d\theta)\\
&  =d[\theta(X_{\phi})]+\omega(X_{\phi},\cdot)\\
&  =d[\theta(X_{\phi})+\phi].
\end{align*}
Thus, if $\mathcal{L}_{X_{\phi}}\theta=0,$ we have $d[\theta(X_{\phi}%
)+\phi]=0.$

For any $\eta\in\mathfrak{k},$ the Hamiltonian vector field $X_{\phi_{\eta}}$
is the generator of the adjoint action of the one-parameter subgroup
$e^{t\eta}$ on $T^{\ast}(K).$ We denote this vector field more compactly as
$X^{\eta}.$ Now, if we take $\phi=\phi_{\eta},$ the action of $X_{\phi_{\eta}%
}=X^{\eta}$ on $T^{\ast}(K)$ is induced from an action on the base. But any
such action will preserve the canonical 1-form $\theta,$ showing that
$\mathcal{L}_{X^{\eta}}\theta=0.$ We conclude, then, that $d[\theta(X^{\eta
})+\phi_{\eta}]=0,$ showing that $\theta(X^{\eta})+\phi_{\eta}$ is a constant
for each \thinspace$\eta\in\mathfrak{k}.$ But since $\theta(X^{\eta})$ and
$\phi_{\eta}$ are easily seen to be zero on the zero section inside $T^{\ast
}(K)$, we actually have $\theta(X^{\eta})+\phi_{\eta}=0.$ Thus, we have
simply
\begin{equation}
Q_{\mathrm{pre}}(\phi_{\eta})=i\hbar X^{\eta}. \label{QpreTheta}%
\end{equation}

Finally, we claim that the form $\beta$ is invariant under the adjoint action
of $K.$ To see this, observe that if we transform $\beta$ by the adjoint
action of $x\in K,$ the resulting form $x\cdot\beta$ will still be a
left-invariant holomorphic $n$-form, which must agree with $\beta$ up to a
constant. It thus suffices to compare $x\cdot\beta$ to $\beta$ at the
identity. But at the identity, $x\cdot\beta=\det(\mathrm{Ad}_{x})\beta=\beta,$
because $\mathrm{Ad}_{x}\in SO(\mathfrak{k})\subset SO(\mathfrak{k}%
_{\mathbb{C}};\mathbb{C}).$ Since the function $e^{-\left\vert \xi\right\vert
^{2}/(2\hbar)}$ is also invariant under the adjoint action of $K,$ we obtain%
\[
Q(\phi_{\eta})\psi=(X^{\eta}F)e^{-\left\vert \xi\right\vert ^{2}/(2\hbar
)}\otimes\sqrt{\beta}.
\]
Thus, the invariant elements are those for which $X^{\eta}F=0,$ i.e., those
invariant under the adjoint action of $K.$ But since $F$ is holomorphic, if
$F$ is invariant under the adjoint action of $K,$ it is also invariant under
the adjoint action of $K_{\mathbb{C}}$; that is, $F$ is a holomorphic class
function on $K_{\mathbb{C}}.$
\end{proof}

\subsection{The theorem of Florentino, Mour\~{a}o, and Nunes\label{fmnThm.sec}%
}

The goal of this section is to describe a formula, obtained by Florentino,
Mour\~{a}o, and Nunes in \cite{FMN}, for computing the $L^{2}$ norm of a
holomorphic class function $F$ with respect to the measure $\gamma_{\hbar}$ in
(\ref{gamma}). The formula expresses the square of the $L^{2}$ of $F$ as a
certain integral of $\left\vert F\right\vert ^{2}$ over $T_{\mathbb{C}}.$ Now,
almost every point in $K_{\mathbb{C}}$ is conjugate to a point---unique up to
the action of $W$---in $T_{\mathbb{C}}.$ It is therefore easy to show that
there is \textit{some} $W$-invariant measure $\mu_{\hbar}$ on $T_{\mathbb{C}}$
such that
\[
\int_{K_{\mathbb{C}}}\left\vert F(g)\right\vert ^{2}~d\gamma_{\hbar}%
=\int_{T_{\mathbb{C}}}\left\vert F(z)\right\vert ^{2}d\mu_{\hbar}(z),
\]
for all functions $F$ (not necessarily holomorphic). What is not obvious is
whether there is any way to compute $\mu_{\hbar}$ explicitly.

To describe the result of Florentino, Mour\~{a}o, and Nunes, we make use of
Proposition \ref{geoHeat.prop}, which relates the measure $\gamma_{\hbar}$ in
(\ref{gamma}) to a heat kernel measure on $K_{\mathbb{C}}.$ We fix a Haar
measure $dg$ on $K_{\mathbb{C}}$ and consider the \textquotedblleft%
$K$-averaged heat kernel\textquotedblright\ $\nu_{\hbar}$ on $K_{\mathbb{C}}$
\cite[Theorem 2]{Hall94}, normalized so that $\nu_{\hbar}(g)~dg$ is a
probability measure. (This measure is just the heat kernel measure for the
noncompact symmetric space $K_{\mathbb{C}}/K,$ viewed as a $K$-invariant
measure on $K_{\mathbb{C}}.$) We let $dz$ and $\nu_{\hbar}^{\prime}$ be the
analogous objects on $T_{\mathbb{C}}.$

\begin{theorem}
[Florentino, Mour\~{a}o, and Nunes]\label{fmn.thm}If $F$ is a holomorphic
class function on $K_{\mathbb{C}}$, then%
\begin{equation}
\int_{K_{\mathbb{C}}}\left\vert F(g)\right\vert ^{2}\nu_{\hbar}(g)~dg=\frac
{e^{-\hbar\left\Vert \delta\right\Vert ^{2}}}{\left\vert W\right\vert }%
\int_{T_{\mathbb{C}}}\left\vert \sigma_{\mathbb{C}}(z)F(z)\right\vert ^{2}%
\nu_{\hbar}^{\prime}(z)~dz, \label{fmnFormula}%
\end{equation}
where $\delta$ is half the sum of the positive roots. Furthermore, if
$\Phi:T_{\mathbb{C}}\rightarrow\mathbb{C}$ is a $W$-alternating holomorphic
function for which%
\[
\int_{T_{\mathbb{C}}}\left\vert \Phi(z)\right\vert ^{2}\nu_{\hbar}^{\prime
}(z)~dz<\infty,
\]
then there exists a unique holomorphic class function $F$ on $K_{\mathbb{C}}$
that is square integrable with respect to $\nu_{\hbar}$ and such that
\[
(\sigma_{\mathbb{C}})(\left.  F\right\vert _{T_{\mathbb{C}}})=\Phi.
\]

\end{theorem}

Recall from Proposition \ref{geoHeat.prop} that the measure $\nu_{\hbar
}(g)~dg$ coincides up to a constant with the measure $\gamma_{\hbar}$ on
$T^{\ast}(K)\cong K_{\mathbb{C}}$ that arises in geometric quantization. (See
(\ref{gamma}).) A similar statement applies to the measure $\nu_{\hbar
}^{\prime}(z)~dz$ on $T^{\ast}(T)\cong T_{\mathbb{C}}.$ Thus, Theorem
\ref{fmn.thm} gives us an \textit{explicit} way of computing the norm of an
invariant element of the quantum Hilbert space as an integral over
$T_{\mathbb{C}}.$ Note also the key role played by the analytically continued
Weyl denominator $\sigma_{\mathbb{C}}$: The integral on the right-hand side of
(\ref{fmnFormula}) is computing the square of the $L^{2}$ norm of
$(\sigma_{\mathbb{C}})(\left.  F\right\vert _{T_{\mathbb{C}}})$---rather than
the norm of $\left.  F\right\vert _{T_{\mathbb{C}}}$---with respect to
$\nu_{\hbar}^{\prime}(z)~dz.$

As we have discussed in Section \ref{main.sec}, the main goal of this paper is
to interpret Theorem \ref{fmn.thm} as a unitary \textquotedblleft quantization
commutes with reduction\textquotedblright\ result. To achieve this goal, we
must (1) show that the quantization of $K_{\mathbb{C}}/\!\!/\mathrm{Ad}_{K}$
can be identified with a space of holomorphic functions on $T_{\mathbb{C}},$
and (2) show that the map%
\[
F\mapsto(\sigma_{\mathbb{C}})(\left.  F\right\vert _{T_{\mathbb{C}}})
\]
arises from a Guillemin--Sternberg-type map with half-forms, similar to the
one in \cite{HallKirwin}.

Theorem \ref{fmn.thm} result is remarkably similar to the Weyl integral
formula (e.g., \cite[Proposition 12.24]{LieBook}), which states that if $f$ is
a continuous class function on $K,$ we have%
\[
\int_{K}\left\vert f(x)\right\vert ^{2}~dx=\frac{1}{\left\vert W\right\vert
}\int_{T}\left\vert \sigma(t)f(t)\right\vert ^{2}~dt,
\]
where $dx$ and $dt$ are the normalized Haar measures on $K$ and $T,$
respectively. In passing from $K$ and $T$ to $K_{\mathbb{C}}$ and
$T_{\mathbb{C}},$ we merely replace the Haar measures with heat kernel
measures, change $\sigma$ to $\sigma_{\mathbb{C}},$ and add a factor of
$e^{-\hbar\left\Vert \delta\right\Vert ^{2}}$ on the right-hand side.

A result similar to Theorem \ref{fmn.thm} for holomorphic functions in the
dual noncompact setting was given by Hall and Mitchell. (Compare the isometry
theorem in \cite{HMcomplex} in the general case to the isometry theorem in
\cite{HMradial} in the radial case.)

Florentino, Mour\~{a}o, and Nunes give two proofs of Theorem \ref{fmn.thm},
one of which (proof of Theorem 2.3 in \cite{FMN}) actually applies to an
arbitrary measurable (i.e., not necessarily holomorphic) class function $F.$
Since the holomorphic case of Theorem \ref{fmn.thm} is a vital result for this
paper, we outline the proof of this case, following \cite{FMN}. Note that the
statements in \cite{FMN} differ by various factors of $2\pi$ from our
statement of Theorem \ref{fmn.thm}, because of differences in the scaling of
the heat equation.

The ingredients of the proof are the generalized Segal--Bargmann transform for
the group $K$ (see \cite{Hall94}), the analogous transform for $T,$ and the
Weyl integral formula. The Segal--Bargmann transform for $K$ is the map
$C_{\hbar}:L^{2}(K)\rightarrow\mathcal{H}(K_{\mathbb{C}})$ given by%
\[
C_{\hbar}(f)=(e^{\hbar\Delta_{K}/2}f)_{\mathbb{C}},
\]
where $\Delta_{K}$ is the (negative) Laplacian for $K,$ $e^{\hbar\Delta_{K}%
/2}$ is the associated heat operator, and $(\cdot)_{\mathbb{C}}$ denotes
analytic continuation from $K$ to $K_{\mathbb{C}}.$ The transform is a unitary
map from $L^{2}(K,dx)$ onto $\mathcal{H}L^{2}(K_{\mathbb{C}},\nu_{\hbar
}(g)~dg),$ where $dx$ is the normalized Haar measure on $K$ and where
$\mathcal{H}L^{2}$ denotes the space of square integrable holomorphic
functions \cite[Theorem 2]{Hall94}. Since $T$ is also a connected compact Lie
group, there is a similar unitary map from $L^{2}(T,dt)$ to $\mathcal{H}%
L^{2}(T_{\mathbb{C}},\nu_{\hbar}^{\prime}(z)~dz).$

\begin{proof}
For each dominant integral element $\mu$, let $\chi_{\mu}:K\rightarrow
\mathbb{C}$ denote the character of the irreducible representation of $K$ with
highest weight $\mu.$ Then $\chi_{\mu}$ has a holomorphic extension to
$K_{\mathbb{C}},$ denoted $(\chi_{\mu})_{\mathbb{C}}.$ We consider first the
case that $F=(\chi_{\mu})_{\mathbb{C}}.$ Now, the character $\chi_{\mu}$
satisfies%
\begin{equation}
\Delta_{K}(\chi_{\mu})=-(\left\Vert \mu+\delta\right\Vert ^{2}-\left\Vert
\delta\right\Vert ^{2})\chi_{\mu}. \label{Delta_Kev}%
\end{equation}
(This claim follows easily from Proposition 10.6 in \cite{LieBook}.) Thus, if
we take $f=e^{\frac{\hbar}{2}(\left\Vert \mu+\delta\right\Vert ^{2}-\left\Vert
\delta\right\Vert ^{2})}\chi_{\lambda},$ we will have $C_{\hbar}(f)=F.$ By the
isometricity of $C_{\hbar}$ and the Weyl integral formula, we then have%
\begin{equation}
\int_{K_{\mathbb{C}}}\left\vert F(g)\right\vert ^{2}\nu_{\hbar}(g)~dg=\int%
_{K}\left\vert f(x)\right\vert ^{2}~dx=\frac{e^{\hbar(\left\Vert \mu
+\delta\right\Vert ^{2}-\left\Vert \delta\right\Vert ^{2})}}{\left\vert
W\right\vert }\int_{T}\left\vert \sigma\chi_{\mu}(t)\right\vert ^{2}~dt,
\label{FMNStep1}%
\end{equation}
where $dx$ and $dt$ are the normalized Haar measures on $K$ and $T,$ respectively.

Meanwhile, the function $\sigma\chi_{\mu}$ on $T$ satisfies%
\begin{equation}
\Delta_{T}(\sigma\chi_{\mu})=-\left\Vert \mu+\delta\right\Vert ^{2}\chi_{\mu
}\text{;} \label{Delta_Tev}%
\end{equation}
note the shift in the eigenvalue between (\ref{Delta_Kev}) and
(\ref{Delta_Tev}). This claim follows from the special form of the
\textquotedblleft radial part\textquotedblright\ of the Laplacian on a compact
Lie group. (See Proposition 2.3 on p. 278 of \cite{Ber}; the proof is
essentially the same as the proof of Proposition 3.10 in Chapter II of
\cite{Hel} in the dual noncompact setting.) But (\ref{Delta_Tev}) also follows
easily from the Weyl character formula: The numerator in the character formula
is easily seen to be an eigenfunction of $\Delta_{T}$ with the stated
eigenvalue. The isometricity of the Segal--Bargmann transform for $T$ then
tells us that
\begin{align}
&  \frac{e^{\hbar(\left\Vert \mu+\delta\right\Vert ^{2}-\left\Vert
\delta\right\Vert ^{2})}}{\left\vert W\right\vert }\int_{T}\left\vert
\sigma\chi_{\mu}(t)\right\vert ^{2}~dt\nonumber\\
&  =\frac{e^{-\hbar\left\Vert \delta\right\Vert ^{2}}}{\left\vert W\right\vert
}\int_{T_{\mathbb{C}}}\left\vert \sigma_{\mathbb{C}}(z)\right\vert
^{2}\left\vert (\chi_{\mu})_{\mathbb{C}}(z)\right\vert ^{2}\nu_{\hbar}%
^{\prime}(z)~dz. \label{FMNStep2}%
\end{align}
Combining (\ref{FMNStep1}) and (\ref{FMNStep2}) establishes (\ref{fmnFormula})
when $F=(\chi_{\mu})_{\mathbb{C}}.$

Now, it follows from the \textquotedblleft holomorphic Peter--Weyl
theorem\textquotedblright\ of \cite[Theorem 9]{Hall94} that the functions
$(\chi_{\mu})_{\mathbb{C}}$ form an orthogonal basis for the space of
holomorphic class functions in $L^{2}(K_{\mathbb{C}},\nu_{\hbar}(g)~dg).$
Meanwhile, it is not hard to show that the functions $(\chi_{\mu}%
)_{\mathbb{C}}$ are also orthogonal in $L^{2}(T_{\mathbb{C}},\nu_{\hbar
}^{\prime}(z)~dz).$ (Using the Segal--Bargmann transform for $T$ along with
(\ref{Delta_Tev}), the desired result reduces to the orthogonality of the
functions $\sigma\chi_{\mu}$ in $L^{2}(T,dt),$ which is a consequence of the
Weyl character formula.) Thus, the general version of (\ref{fmnFormula})
reduces to the already established case for characters.

Finally, suppose $\Phi:T_{\mathbb{C}}\rightarrow\mathbb{C}$ is as in the
second part of the theorem. We can expand $\Phi$ in a Fourier--Laurent series
in terms of the exponential functions%
\[
f_{\lambda}(e^{H})=e^{i\left\langle \lambda,H\right\rangle },\quad
H\in\mathfrak{t}_{\mathbb{C}},
\]
where $\lambda$ ranges over all integral elements in $\mathfrak{t}.$ (If we
identify $T_{\mathbb{C}}$ with $(\mathbb{C}^{\ast})^{r},$ these functions are
just the monomials.) If $\Phi$ is $W$-alternating, the coefficients in the
expansion of $\Phi$ must also be $W$-alternating. Thus, the coefficient of
$f_{\lambda}$ will be zero if $\lambda$ belongs to any of the walls of the
Weyl chambers. The coefficients where $\lambda$ is not in the wall of any
chamber, meanwhile, can be grouped into Weyl orbits. If $\lambda$ is in the
interior of the fundamental Weyl chamber, then $\lambda=\mu+\delta$ for some
$\mu$ in the closed fundamental Weyl chamber \cite[Proposition 8.38]{LieBook}.
The group of exponentials coming from the Weyl-orbit of $\lambda$ is then the
numerator in the Weyl character formula for the representation with highest
weight $\mu.$ The desired $F$ can then be constructed as a linear combination
of the analytically continued characters $(\chi_{\mu}),$ with the isometricity
in (\ref{fmnFormula}) guaranteeing convergence of the expansion.
\end{proof}

\section{Reduction of the classical phase space}

Recall from (\ref{TTInTK}) that we think of $T^{\ast}(T)$ as a submanifold of
$T^{\ast}(K).$ We are are going to identify a \textquotedblleft regular
set\textquotedblright\ $\phi^{-1}(0)^{\mathrm{reg}}$ inside the zero set
$\phi^{-1}(0)$ of the momentum map. We are mainly interested in the regular
part of the reduced phase space,%
\[
\phi^{-1}(0)^{\mathrm{reg}}/\mathrm{Ad}_{K},
\]
which is referred to as the \textquotedblleft principal
stratum\textquotedblright\ in \cite{HRS} and \cite{BLV}. We will see that
$T^{\ast}(T)$ is contained in $\phi^{-1}(0);$ we then define $T^{\ast
}(T)^{\mathrm{reg}}$ as the intersection of $T^{\ast}(T)$ with $\phi
^{-1}(0)^{\mathrm{reg}}.$ We will show that the regular part of the reduced
phase space is a smooth symplectic manifold, which may be identified as%
\[
\phi^{-1}(0)^{\mathrm{reg}}/\mathrm{Ad}_{K}=T^{\ast}(T)^{\mathrm{reg}}/W.
\]
In addition, we will show that $\phi^{-1}(0)^{\mathrm{reg}}/\mathrm{Ad}_{K}$
inherits a K\"{a}hler structure from the K\"{a}hler structure on $T^{\ast
}(K)\cong K_{\mathbb{C}}.$ As a complex manifold, we have%
\[
\phi^{-1}(0)^{\mathrm{reg}}/\mathrm{Ad}_{K}\cong T_{\mathbb{C}}^{\mathrm{rss}%
}/W,
\]
where $T_{\mathbb{C}}^{\mathrm{rss}}$ denotes the set of \textquotedblleft
regular semisimple\textquotedblright\ points in $T_{\mathbb{C}}.$ The reader
who wishes to take these identifications on faith may look at the statements
of Theorems \ref{regularPoints.thm} and \ref{quotient.thm} and then proceed to
Section \ref{quantReduced.sec}.

Although some of the calculations in this section have appeared elsewhere
(e.g., Section 1 of \cite{Hueb} or Section 4.1 of \cite{BLV}), we give special
emphasis to identifying the regular set and it is therefore simplest to give
complete proofs.

\subsection{The zero set of the momentum map}

Recall the formula for the momentum map $\phi:T^{\ast}(K)\rightarrow
\mathfrak{k}^{\ast}$ given in (\ref{momentum}) in Section
\ref{momentumMap.sec}. From the formula, we immediately obtain that the
zero-set of $\phi$ is as follows:%
\begin{equation}
\phi^{-1}(0)=\left\{  \left.  (x,\xi)\in T^{\ast}(K)\right\vert \mathrm{Ad}%
_{x}(\xi)=\xi\right\}  . \label{zeroset}%
\end{equation}
Recall also that we identify $T^{\ast}(T)$ as a subset of $T^{\ast}(K)$ as in
(\ref{TTInTK}).

\begin{proposition}
\label{zeroSet.prop}Every point in $T^{\ast}(T)$ belongs to $\phi^{-1}(0)$ and
each $\mathrm{Ad}_{K}$-orbit in $\phi^{-1}(0)$ intersects $T^{\ast}(T)$ in
exactly one $W$-orbit.
\end{proposition}

\begin{proof}
First, since $T$ is commutative, every point in $T^{\ast}(T)$ certainly
satisfies the condition in (\ref{zeroset}). Second, suppose that $(x,\xi
)\in\phi^{-1}(0).$ Then $x$ commutes with every element of the connected,
commutative subgroup $S:=\{e^{t\xi}\}_{t\in\mathbb{R}}$ of $K.$ Thus, by Lemma
11.37 of \cite{LieBook}, there is a maximal torus $S^{\prime}$ that contains
both $x$ and $S.$ By the torus theorem, $S^{\prime}$ is conjugate to $T$.
Thus, there is some $y\in K$ such that $h:=yxy^{-1}$ belongs to $T$ and
$H:=y\xi y^{-1}$ belongs to $\mathfrak{t},$ showing that $(x,\xi)$ can be
moved to a point in $T^{\ast}(T).$

Last, suppose $(t,H)$ and $(t^{\prime},H^{\prime})$ in $T^{\ast}(T)$ belong to
the same $\mathrm{Ad}_{K}$-orbit. A standard result in the theory of compact
groups says that if two elements of $T$ are conjugate in $K,$ they belong to
the same Weyl group orbit. The same proof applies without change here to show
that $(t,H)$ and $(t^{\prime},H^{\prime})$ must be in the same Weyl group
orbit. In the proof of Theorem 11.39 in \cite{LieBook}, for example, we may
simply replace the centralizer of $t$ by the stabilizer of $(t,H)$ and the
argument goes through without change.
\end{proof}

\subsection{Regular points\label{regularPoints.sec}}

The action of $K$ on $\phi^{-1}(0)$ is not even generically free. We can
nevertheless identify a \textquotedblleft regular set\textquotedblright\ in
$\phi^{-1}(0)$ where the stabilizer is as small as possible. Define, for each
$(x,\xi)\in\phi^{-1}(0),$ the \textbf{stabilizer} $S_{(x,\xi)}$ as
\[
S_{(x,\xi)}=\left\{  \left.  y\in K\right\vert yxy^{-1}=x,~\mathrm{Ad}_{y}%
(\xi)=\xi\right\}  .
\]
It follows from Proposition \ref{zeroSet.prop} that for all $(x,\xi)\in
\phi^{-1}(0),$ the stabilizer of $(x,\xi)$ \textit{contains} a maximal torus
in $K.$

\begin{definition}
A point $(x,\xi)$ in $\phi^{-1}(0)$ is called \textbf{regular} if $S_{(x,\xi
)}$ \emph{is} a maximal torus in $K.$
\end{definition}

The set of regular points is referred to as the \textquotedblleft principal
stratum\textquotedblright\ in \cite{HRS} and \cite{BLV}. The other strata in
those papers are defined by specifying the conjugacy class of the stabilizer.

We would like to understand when a point in $\phi^{-1}(0)$ is regular. In
light of Proposition \ref{zeroSet.prop}, it suffices to consider points in
$T^{\ast}(T).$

\begin{theorem}
\label{regularPoints.thm}Consider a point $(e^{H_{1}},H_{2})$ in $T^{\ast
}(T)\subset\phi^{-1}(0).$ Then $(e^{H_{1}},H_{2})$ is regular if and only if
for each root $\alpha,$ we have \emph{either}
\[
\left\langle \alpha,H_{1}\right\rangle \notin2\pi\mathbb{Z}%
\]
or
\[
\left\langle \alpha,H_{2}\right\rangle \neq0.
\]

\end{theorem}

It should be emphasized that this characterization of the regular set in
$\phi^{-1}(0)$ is valid only because of our standing assumption that $K$ is
simply connected. (This assumption is used in the proof of Proposition
\ref{rss.prop}.) Note that the set of regular points in $T^{\ast}(T)$ is open
and dense in $T^{\ast}(T).$

\begin{corollary}
\label{openDense.cor}The set of regular points in $\phi^{-1}(0)$ is open and dense.
\end{corollary}

\begin{proof}
Given $(x,\xi)\in\phi^{-1}(0),$ we can find (Proposition \ref{zeroSet.prop})
some $y\in K,$ some $t\in T,$ and some $H\in\mathfrak{t}$ such that
$(x,\xi)=y\cdot(t,H).$ We can then find some $(t^{\prime},H^{\prime})$ very
near $(t,H)$ in $T^{\ast}(T)$ that satisfies the condition in Theorem
\ref{regularPoints.thm}. Thus, $y\cdot(t^{\prime},H^{\prime})$ is a regular
point in $\phi^{-1}(0)$ very near to $(x,\xi),$ showing that the regular set
is dense.

Suppose now that $(x,\xi)\in\phi^{-1}(0)$ is regular and that $(x_{n},\xi
_{n})$ is a sequence in $\phi^{-1}(0)$ converging to $(x,\xi).$ Then we can
write $(x_{n},\xi_{n})=y_{n}\cdot(t_{n},H_{n})$ for some $y_{n}\in K,$
$t_{n}\in T,$ and $H_{n}\in\mathfrak{k}.$ Since $\xi_{n}$ is converging to
$\xi,$ there is some constant $C$ such that%
\[
\left\Vert H_{n}\right\Vert =\left\Vert \xi_{n}\right\Vert \leq C.
\]
Thus, using compactness, we can extract convergent sequences and assume that
$y_{n}\rightarrow y,$ $t_{n}\rightarrow t$, and $H_{n}\rightarrow H.$ Then%
\[
(x,\xi)=\lim_{n\rightarrow\infty}(x_{n},\xi_{n})=\lim_{n\rightarrow\infty
}y_{n}\cdot(t_{n},H_{n})=y\cdot(t,H).
\]
Since $(x,\xi)$ is assumed regular, $(t,H)$ must be in the regular set in
$T^{\ast}(T).$ But this set is open in $T^{\ast}(T),$ showing that
$(t_{n},H_{n})$ and therefore also $y_{n}\cdot(t_{n},H_{n})$ are regular for
all sufficiently large $n.$
\end{proof}

We now give the proof of Theorem \ref{regularPoints.thm}, which consists of a
series of propositions.

\begin{proposition}
For each $(x,\xi)\in\phi^{-1}(0),$ the stabilizer $S_{(x,\xi)}$ coincides with
the intersection of the centralizer of $\Psi(x,\xi):=xe^{-i\xi}\in
K_{\mathbb{C}}$ with $K$:%
\[
S_{(x,\xi)}=C_{K}(xe^{-i\xi})=\left\{  \left.  y\in K\right\vert y(xe^{-i\xi
})y^{-1}=xe^{-i\xi}\right\}  .
\]

\end{proposition}

\begin{proof}
Clearly, if $y$ commutes with both $x$ and $\xi,$ then $y$ commutes with
$xe^{-i\xi}.$ Conversely, if $y\in K$ and $y$ commutes with $xe^{-i\xi},$ then
$y$ must commute with both $x$ and $\xi.$ After all,
\[
y(xe^{-i\xi})y^{-1}=(yxy^{-1})e^{-i\mathrm{Ad}_{y}(\xi)}.
\]
By the uniqueness of the polar decomposition, the above quantity equals
$xe^{-i\xi}$ only if $yxy^{-1}=x$ and $\mathrm{Ad}_{y}(\xi)=\xi.$
\end{proof}

\begin{definition}
\label{rss.def}An element $g$ of $K_{\mathbb{C}}$ is called\textbf{ regular
semisimple} if the centralizer of $g$ is a complex maximal torus in
$K_{\mathbb{C}}.$ We denote the set of regular semisimple elements in
$K_{\mathbb{C}}$ by $K_{\mathbb{C}}^{\mathrm{rss}}$ and the set of regular
semisimple elements in $T_{\mathbb{C}}$ by $T_{\mathbb{C}}^{\mathrm{rss}}.$
\end{definition}

Since $K$ and $K_{\mathbb{C}}$ are assumed simply connected, Steinberg's
theorem \cite[Theorem 2.11]{HumConj} says that the centralizer of every
semisimple element is connected. Thus, in our setting, Definition
\ref{rss.def} is equivalent to the usual definition of a regular semisimple
element (e.g., \cite[Section 1.6]{HumConj}).

\begin{proposition}
\label{rss.prop}For each root $\alpha,$ let $\phi_{\alpha}:T_{\mathbb{C}%
}\rightarrow\mathbb{C}^{\ast}$ be the associated root homomorphism given by
\begin{equation}
\phi_{\alpha}(e^{H})=e^{i\left\langle \alpha,H\right\rangle }.
\label{phiAlphaDef}%
\end{equation}
Then $z\in T_{\mathbb{C}}$ is regular semisimple if and only if for all
$\alpha\in R,$ we have
\[
\phi_{\alpha}(z)\neq1.
\]

\end{proposition}

The factor of $i$ in the exponent in the formula for $\phi_{\alpha}$ is a
result of our convention for using real roots \cite[Definition 11.34]{LieBook}.

\begin{proof}
As we have noted, the assumption that $K$ is simply connected ensures that our
notion of a regular semisimple element is equivalent to the one in
\cite{HumConj}. The result then follows immediately from Proposition 2.3 in
\cite{HumConj}.
\end{proof}

\begin{proposition}
\label{regularRSS.prop}A point $(x,\xi)\in\phi^{-1}(0)$ is regular if and only
if the element $\Psi(x,\xi)=xe^{-i\xi}$ in $K_{\mathbb{C}}$ is regular semisimple.
\end{proposition}

\begin{proof}
In light of Proposition \ref{zeroSet.prop}, it is harmless to assume that
$x\in T$ and $\xi\in\mathfrak{t},$ so that $xe^{-i\xi}\in T_{\mathbb{C}}.$ We
will prove that $(x,\xi)$ fails to be regular if and only if $xe^{-i\xi}$
fails to be regular semisimple. Suppose first that $(x,\xi)$ fails to be
regular. Then the stabilizer of $(x,\xi)$ contains an element $y\in K$ that is
not in $T.$ Then $y$ is not in the complexification $T_{\mathbb{C}}$ of $T,$
since $T_{\mathbb{C}}\cap K=T.$ Then the centralizer of $xe^{-i\xi}$ contains
$y$ and thus properly contains the complex maximal torus $T_{\mathbb{C}},$
showing that $xe^{-i\xi}$ is not regular semisimple.

Suppose now that $z:=xe^{-i\xi}\in T_{\mathbb{C}}$ fails to be regular
semisimple. Then by Proposition \ref{rss.prop}, $z$ belongs to the kernel of
$\phi_{\alpha}$ for some $\alpha,$ and therefore also to the kernel of
$\phi_{-\alpha}=1/\phi_{\alpha}.$ But then if $X$ is in the root space
$\mathfrak{(k}_{\mathbb{C}})_{\alpha}$ we have
\[
\mathrm{Ad}_{z}(X)=\phi_{\alpha}(z)X=X
\]
and similarly if $X\in(\mathfrak{k}_{\mathbb{C}})_{-\alpha}.$ Thus, the Lie
algebra of the centralizer of $z$ contains $(\mathfrak{k}_{\mathbb{C}%
})_{\alpha}\oplus(\mathfrak{k}_{\mathbb{C}})_{-\alpha},$ which contains
elements of $\mathfrak{k}$ not in $\mathfrak{t}$ \cite[Corollary
7.20]{LieBook}.
\end{proof}

We are now in a position to complete the proof of Theorem
\ref{regularPoints.thm}. By Proposition \ref{regularRSS.prop}, a point $(t,H)$
in $T^{\ast}(T)\subset\phi^{-1}(0)$ is regular if and only if the
corresponding point $z=xe^{-i\xi}$ in $T_{\mathbb{C}}$ is regular semisimple,
which holds (Proposition \ref{rss.prop}) if and only if $\phi_{\alpha}(z)$ is
different from 1 for all $\alpha.$ But%
\[
\phi_{\alpha}(e^{H_{1}}e^{-iH_{2}})=e^{i\left\langle \alpha,H_{1}\right\rangle
}e^{\left\langle \alpha,H_{2}\right\rangle }=1
\]
if and only if $\left\langle \alpha,H_{1}\right\rangle \in2\pi\mathbb{Z}$ and
$\left\langle \alpha,H_{2}\right\rangle =0.$

\subsection{The reduced phase space}

Suppose $(M,\omega)$ is a symplectic manifold and $K$ is a compact Lie group
acting symplectically on $M.$ If the action of $K$ admits an equivariant
momentum map $\phi,$ the \textbf{symplectic quotient} (or Marsden--Weinstein
quotient) $M/\!\!/G$ is defined as the ordinary quotient of $\phi^{-1}(0)$ by
$K$:%
\[
M/\!\!/K:=\phi^{-1}(0)/K.
\]
Suppose, for example, that $M=T^{\ast}(N)$ and the action of $K$ on $T^{\ast
}(N)$ is induced from a regular action of $K$ on $N.$ (The action is called
regular if the stabilizers of any two points are conjugate, for example, if
the action is free.) Then an equivariant momentum map may be constructed that
is linear on each fiber in $T^{\ast}(N),$ and we have%
\begin{equation}
T^{\ast}(N)/\!\!/K\cong T^{\ast}(N/K). \label{cotangentReduce}%
\end{equation}
(See Section 4.3 in \cite{AM} and especially the $\mu=0$ case of Theorem 4.3.3.)

In our case, $M=T^{\ast}(K)$ and $K$ acts on itself---and therefore also on
$T^{\ast}(K)$---by the adjoint action. The adjoint action of $K$ on itself is
not regular, however, and the quotient $K/\mathrm{Ad}_{K}$ is not a manifold.
Rather, $K/\mathrm{Ad}_{K}$ is identified with $T/W,$ which even when
$K=SU(2)$ is a closed interval rather than a smooth manifold. In light of
(\ref{cotangentReduce}), we expect that $T^{\ast}(K)/\!\!/\mathrm{Ad}_{K}$
should be something like $T^{\ast}(T/W).$ Since $T/W$ is not a manifold,
however, the correct statement is that $T^{\ast}(K)/\!\!/\mathrm{Ad}_{K}$ is
homeomorphic to $T^{\ast}(T)/W.$ (This claim follows from Proposition
\ref{zeroSet.prop}.)

In this paper, we will focus on the set $\phi^{-1}(0)^{\mathrm{reg}}$ of
regular points in $\phi^{-1}(0),$ and the quotient $\phi^{-1}(0)^{\mathrm{reg}%
}/\mathrm{Ad}_{K}.$ To describe this quotient, recall that we think of
$T^{\ast}(T)$ as a submanifold of $T^{\ast}(K)$ as in (\ref{TTInTK}) and that
$T^{\ast}(T)$ is contained in $\phi^{-1}(0).$ We define the regular set
$T^{\ast}(T)^{\mathrm{reg}}$ in $T^{\ast}(T)$ as%
\[
T^{\ast}(T)^{\mathrm{reg}}=T^{\ast}(T)\cap\phi^{-1}(0)^{\mathrm{reg}}.
\]
Then $T^{\ast}(T)^{\mathrm{reg}}$ is an open dense subset of $T^{\ast}(T)$ and
the Weyl group acts freely on this set.

\begin{theorem}
\label{quotient.thm}The quotient $\phi^{-1}(0)^{\mathrm{reg}}/\mathrm{Ad}_{K}$
is a smooth manifold, which may be identified as%
\[
\phi^{-1}(0)^{\mathrm{reg}}/\mathrm{Ad}_{K}\cong T^{\ast}(T)^{\mathrm{reg}%
}/W.
\]
This manifold inherits a K\"{a}hler structure from the K\"{a}hler structure on
$T^{\ast}(K)\cong K_{\mathbb{C}}.$ The symplectic structure on the quotient
comes from the canonical symplectic structure on $T^{\ast}(T)$ and the complex
structure on the quotient comes from the complex structure obtained by
identifying $T^{\ast}(T)$ with $T_{\mathbb{C}}.$
\end{theorem}

Recall also that (Proposition \ref{regularRSS.prop}) under the identification
of $T^{\ast}(T)$ with $T_{\mathbb{C}},$ the set of regular points in $T^{\ast
}(T)$ corresponds to the set of regular semisimple points in $T_{\mathbb{C}}.$
Thus, as a complex manifold, we may think of the regular reduced phase space
as%
\begin{equation}
\phi^{-1}(0)^{\mathrm{reg}}/\mathrm{Ad}_{K}\cong T_{\mathbb{C}}^{\mathrm{rss}%
}/W. \label{tcrssW}%
\end{equation}

We now give the proof of Theorem \ref{quotient.thm}, which consists of a
series of propositions.

\begin{proposition}
The set $\phi^{-1}(0)^{\mathrm{reg}}$ is a smooth embedded submanifold of
$T^{\ast}(K).$
\end{proposition}

\begin{proof}
By Proposition \ref{zeroSet.prop}, every element of $\phi^{-1}%
(0)^{\mathrm{reg}}$ can be obtained from a point in $T^{\ast}(T)^{\mathrm{reg}%
}$ by the action of $K,$ and the point in $T^{\ast}(T)^{\mathrm{reg}}$ is
unique up to the action of $W.$ Since the stabilizer of each point in
$T^{\ast}(T)^{\mathrm{reg}}$ is $T,$ we have a smooth surjective map%
\[
f:T^{\ast}(T)^{\mathrm{reg}}\times(K/T)\rightarrow\phi^{-1}(0)^{\mathrm{reg}}%
\]
that is $\left\vert W\right\vert $ to one, given by%
\[
f((t,H),[y])=y\cdot(t,H).
\]
We claim that $f$ has a continuous local inverse. This claim amounts to saying
that for $(x,\xi)$ in a small open set in $\phi^{-1}(0)^{\mathrm{reg}},$ it is
possible to choose the point in $T^{\ast}(T)$ to depend continuously on
$(x,\xi).$ This last point is a standard argument similar to the proof of
Corollary \ref{openDense.cor} and is omitted.

It thus suffices to show that the differential of $f$ is injective at each
point. By the $K$-equivariance of $f,$ it suffices to check this at a point of
the form $((t,H),[e])$ with $(t,H)\in T^{\ast}(T)^{\mathrm{reg}}.$ Let us
identify the tangent space to $T^{\ast}(K)$ at any point with $\mathfrak{k}%
\oplus\mathfrak{k}$ using left translation and the inner product on
$\mathfrak{k}.$ Then the image of $f_{\ast}$ at the point $((t,H),[e])$ is
easily computed to consist of vectors of the form%
\begin{equation}
(H_{1},H_{2})+(\mathrm{Ad}_{t^{-1}}(X)-X,[H,X]),\quad H_{1},H_{2}%
\in\mathfrak{t},~X\in\mathfrak{k}, \label{phi0tangent}%
\end{equation}
where the second term is easily seen to lie in $\mathfrak{t}^{\bot}%
\oplus\mathfrak{t}^{\bot}.$ Now, if
\[
(\mathrm{Ad}_{t^{-1}}(X)-X,[H,X])=(0,0),
\]
then $X$ is in the Lie algebra of $S_{(x,\xi)},$ i.e., $X\in\mathfrak{t}.$ If
follows that the dimension of the image of $f_{\ast}$ equals $2\dim
T+\dim(\mathfrak{k})-\dim(\mathfrak{t}),$ showing that $f_{\ast}$ is injective.
\end{proof}

\begin{proposition}
The action of $K$ on $\phi^{-1}(0)^{\mathrm{reg}}$ set is regular (i.e., all
stabilizers are conjugate). Thus, the quotient is a manifold, which may be
identified as%
\[
\phi^{-1}(0)^{\mathrm{reg}}/\mathrm{Ad}_{K}\cong T^{\ast}(T)^{\mathrm{reg}%
}/W.
\]

\end{proposition}

\begin{proof}
By definition, the stabilizer of every point in $\phi^{-1}(0)^{\mathrm{reg}}$
is a maximal torus in $K,$ and all such tori are conjugate. A general result
then shows that the quotient is a manifold. (In general, the quotient of a
manifold by a compact group action is \textquotedblleft
stratified\textquotedblright\ by manifolds associated to different strata, but
if all stabilizers are conjugate, there is only one stratum. See, for example,
Section 2.7 in \cite{DK}.) By Proposition \ref{zeroSet.prop}, every element of
$\phi^{-1}(0)^{\mathrm{reg}}$ can be moved by the action of $K$ to a point in
$T^{\ast}(T)^{\mathrm{reg}}$ that is unique up to the action of $W,$ giving
the claimed identification of the quotient.
\end{proof}

\begin{proposition}
The quotient manifold $T^{\ast}(T)^{\mathrm{reg}}/W$ inherits a symplectic
structure, which comes from the canonical symplectic structure on $T^{\ast
}(T)^{\mathrm{reg}}\subset T^{\ast}(T).$
\end{proposition}

\begin{proof}
For each point in the quotient, we choose a preimage in $\phi^{-1}%
(0)^{\mathrm{reg}},$ which may be taken to be a point $(t,H)$ in $T^{\ast
}(T)^{\mathrm{reg}}.$ Let $V$ denote the tangent space to $\phi^{-1}%
(0)^{\mathrm{reg}}$ at $(t,H)$ and let $W\subset V$ denote the tangent space
to the $\mathrm{Ad}_{K}$-orbit through $(t,H).$ We then restrict the canonical
2-form $\omega$ on $T^{\ast}(K)$ to $V.$ By an elementary general result,
$\omega(w,v)=0$ for all $w\in W$ and $v\in V.$ (See \cite[Lemma 4.3.2]{AM}.)
Thus, $\omega$ descends to the quotient space $V/W,$ which is just the tangent
space to the reduced manifold $\phi^{-1}(0)^{\mathrm{reg}}/\mathrm{Ad}_{K}.$

We now compute this reduced form. The tangent space to $\phi^{-1}%
(0)^{\mathrm{reg}}$ at $(t,H)$ is the direct sum of the tangent space to
$T^{\ast}(T)$ and the tangent space to $K/T,$ which is just the tangent space
to the $\mathrm{Ad}_{K}$-orbit of $(t,H).$ In light of the just-cited general
result, the symplectic form on the reduced space will be just the restriction
of the canonical 2-form on $T^{\ast}(K)$ to $T^{\ast}(T),$ which is the
canonical 2-form on $T^{\ast}(T).$
\end{proof}

\begin{proposition}
The complex structure on $T^{\ast}(K)\cong K_{\mathbb{C}}$ descends to the
complex structure on $T^{\ast}(T)^{\mathrm{reg}}/W$ given by the
identification of $T^{\ast}(T)^{\mathrm{reg}}$ with $T_{\mathbb{C}%
}^{\mathrm{rss}}.$
\end{proposition}

\begin{proof}
For each point in the quotient, we choose a preimage in $\phi^{-1}%
(0)^{\mathrm{reg}},$ which may taken to be a point $(t,H)$ in $T^{\ast
}(T)^{\mathrm{reg}}.$ Let $V$ denote the tangent space to $\phi^{-1}%
(0)^{\mathrm{reg}}$ at $(t,H)$ and let $U$ denote the space of vectors $X$ in
$V$ for which $JX$ is also in $V.$ (Here $J$ is the complex structure on
$T^{\ast}(K)\cong K_{\mathbb{C}}.$) It is not hard to compute that $U$ is just
the tangent space to $T^{\ast}(T)$ and thus that $V$ is the direct sum of $U$
and the tangent space to the $\mathrm{Ad}_{K}$-orbit through $(t,H).$ Thus,
the restriction of $J$ to $U$ descends to a map on the quotient, which is just
the complex structure on $T^{\ast}(T)\cong T_{\mathbb{C}}.$
\end{proof}

\begin{remark}
Let us regard $\phi^{-1}(0)$ as a subset of $K_{\mathbb{C}}$ by means of the
identification of $T^{\ast}(K)$ with $K_{\mathbb{C}}.$ It is then not hard to
show that if $g\in K_{\mathbb{C}}^{\mathrm{rss}},$ the conjugacy class of $g$
intersects $\phi^{-1}(0)^{\mathrm{reg}}$ in exactly one $\mathrm{Ad}_{K}%
$-orbit. We thus have an alternative characterization of the regular part of
the reduced phase space as%
\[
\phi^{-1}(0)^{\mathrm{reg}}/\mathrm{Ad}_{K}\cong K_{\mathbb{C}}^{\mathrm{rss}%
}/\mathrm{Ad}_{K_{\mathbb{C}}}.
\]

\end{remark}

Although we will not use this result in what follows, it is an illuminating
way of thinking about the complex structure on the reduced phase space.

\section{Quantization of the reduced phase space\label{quantReduced.sec}}

\subsection{Quantization of $T_{\mathbb{C}}$}

The reduced phase space is a quotient of an open, dense subset of
$T_{\mathbb{C}}$ by the action of the Weyl group. It is therefore natural to
first consider the quantization of $T_{\mathbb{C}}.$ Since $T$ is a compact,
connected Lie group, we can (and do) quantize $T_{\mathbb{C}}\cong T^{\ast
}(K)$ the same way we quantized $K_{\mathbb{C}}\cong T^{\ast}(K).$ Elements of
$\mathrm{Quant}(T_{\mathbb{C}})$ have the form%
\[
\psi=Fe^{-\left\vert H\right\vert ^{2}/(2\hbar)}\otimes\sqrt{\beta^{\prime}},
\]
where $\beta^{\prime}$ is a nowhere-vanishing, invariant holomorphic $r$-form
and $F$ is a holomorphic function on $T_{\mathbb{C}}$. The norm of $\psi$ is
the $L^{2}$ norm of $F$ with respect to the measure
\begin{equation}
\gamma_{\hbar}^{\prime}:=e^{-\left\vert H\right\vert ^{2}/\hbar}\eta^{\prime
}\varepsilon^{\prime} \label{gammaPrime}%
\end{equation}
where $\varepsilon^{\prime}$ is the symplectic volume measure on $T^{\ast}(T)$
and $\eta^{\prime}$ is defined, analogously to (\ref{hermitianK12}), as
\begin{equation}
\eta^{\prime}:=\left[  \frac{\beta^{\prime}\wedge\bar{\beta}^{\prime}%
}{b^{\prime}\varepsilon^{\prime}}\right]  ^{1/2}. \label{etaPrime}%
\end{equation}
Here the quantity \textquotedblleft$H$\textquotedblright\ on $T_{\mathbb{C}}$
is defined by means of the identification of $T_{\mathbb{C}}$ with $T^{\ast
}(T).$ Actually, since $T_{\mathbb{C}}$ is commutative, the function
$\eta^{\prime}$ is easily seen to be constant. We follow the notational
convention of using primes to distinguish constructs on $T_{\mathbb{C}}$ from
their counterparts on $K_{\mathbb{C}}.$

Since we are going to quotient (an open dense subset of) $T_{\mathbb{C}}$ by
$W,$ it is natural to look for subspaces of the above Hilbert space with
particular transformation properties under $W.$ The difficulty with this idea
is that $\beta^{\prime}$ is not invariant under the action of $W,$ but rather
transforms according to the sign of the Weyl-group element:%
\[
w\cdot\beta^{\prime}=\mathrm{sign}(w)\beta^{\prime}.
\]
Thus, the the most natural way for for $W$ to act on the Hilbert space is by
the following \textit{projective} unitary action%
\begin{equation}
U(w)\left(  F(z)e^{-\left\vert H\right\vert ^{2}/(2\hbar)}\otimes\sqrt
{\beta^{\prime}}\right)  =\sqrt{\mathrm{sign}(w)}F(w^{-1}\cdot
z)e^{-\left\vert H\right\vert ^{2}/(2\hbar)}\otimes\sqrt{\beta^{\prime}}.
\label{wAction}%
\end{equation}
Here we allow both possible signs for the square root of $\sqrt{\mathrm{sign}%
(w)},$ so that $U(w)$ is actually a pair of unitary operators differing by a
sign. These operators satisfy (for any choice of the signs involved)
\[
U(w_{1}w_{2})=\pm U(w_{1})U(w_{2}).
\]

Now, if $\alpha$ is a root and $s_{\alpha}$ is the associated reflection, then
for either choice of the sign in the definition, we have $U(s_{\alpha}%
)^{2}=-I.$ Thus, there are, strictly speaking, no nonzero \textquotedblleft
Weyl-invariant\textquotedblright\ elements in the Hilbert space! Nevertheless,
we can make the following definition.

\begin{definition}
\label{alternating.def}For each root $\alpha,$ let $s_{\alpha}$ be the
associated reflection. Let us choose a sign for the operator $U(s_{\alpha})$
by choosing the factor of $\sqrt{\mathrm{sign}(s_{\alpha})}=\sqrt{-1}$ to have
the value%
\begin{equation}
\sqrt{\mathrm{sign}(s_{\alpha})}=i. \label{signChoice}%
\end{equation}
For each $\psi$ in the Hilbert space, we say that $\psi$ is
\textbf{Weyl-invariant} if
\[
U(s_{\alpha})\psi=i\psi,\quad\forall\alpha\in R,
\]
and we say $\psi$ is \textbf{Weyl-alternating} if%
\[
U(s_{\alpha})\psi=-i\psi,\quad\forall\alpha\in R.
\]

\end{definition}

Of course, Definition \ref{alternating.def} is just a fancy way of saying that
if
\[
\psi=Fe^{-\left\vert H\right\vert ^{2}/(2\hbar)}\otimes\sqrt{\beta^{\prime}},
\]
then $\psi$ is Weyl invariant if $F$ is Weyl invariant and $\psi$ is Weyl
alternating if $F$ is Weyl alternating. Note that the Weyl-invariant and
Weyl-alternating elements have very different behavior as we approach the
singular points.

\begin{proposition}
\label{holoExtend.prop}Suppose $F$ is a holomorphic function on $T_{\mathbb{C}%
}^{\mathrm{rss}}$ that is square integrable with respect to the measure
$\gamma_{\hbar}^{\prime}$ in (\ref{gammaPrime}). Then $F$ has a unique
holomorphic extension to $T_{\mathbb{C}}.$
\end{proposition}

There certainly exist holomorphic functions on $T_{\mathbb{C}}^{\mathrm{rss}}$
that do not extend holomorphically to $T_{\mathbb{C}},$ such as the reciprocal
of the analytically continued Weyl denominator $\sigma_{\mathbb{C}}.$ We are
claiming, however, that such functions cannot be square integrable.

\begin{proof}
The set of irregular points is a complex analytic subvariety of $T_{\mathbb{C}%
}$ defined by the vanishing of the holomorphic function%
\[
F(z)=\prod_{\alpha\in R^{+}}(\phi_{\alpha}(z)-1),
\]
where $\phi_{\alpha}:T_{\mathbb{C}}\rightarrow\mathbb{C}^{\ast}$ is the root
homomorphism in (\ref{phiAlphaDef}). The result then follows from a standard
removable singularities theorem for square-integrable holomorphic functions,
such as Theorem 1.13 and Proposition 1.14 in \cite{Ohsawa}.
\end{proof}

\subsection{ Quantization of $T_{\mathbb{C}}^{\mathrm{rss}}/W$%
\label{reducedQuant.sec}}

Recall that the full reduced phase space $\phi^{-1}(0)/\mathrm{Ad}_{K}$ is not
a manifold. We deal with this difficulty in a simple way, by quantizing only
the set of regular points, $\phi^{-1}(0)^{\mathrm{reg}}/\mathrm{Ad}_{K},$
which we have identified with $T_{\mathbb{C}}^{\mathrm{rss}}/W.$ Our
justification for ignoring the singular points is that we will quantize
$T_{\mathbb{C}}^{\mathrm{rss}}/W$ in such a way that elements of the
quantization may be identified as $W$-alternating holomorphic functions on
$T_{\mathbb{C}}^{\mathrm{rss}}$ that are square integrable with respect to the
measure $\gamma_{\hbar}^{\prime}.$ By Proposition \ref{holoExtend.prop}, every
such function extends holomorphically to all of $T_{\mathbb{C}}.$ Suppose
there were some quantization of the full reduced phase space. It is not clear
what nicer properties an element of this quantization should have than those
already possessed by elements of the quantization of the regular set.

As just mentioned, we are going to quantize $T_{\mathbb{C}}^{\mathrm{rss}}/W$
in such a way that the Hilbert space corresponds to the space of\textit{
Weyl-alternating elements of the quantization of }$T_{\mathbb{C}}.$ In
practical terms, we \textquotedblleft need\textquotedblright\ this to be the
case, in the sense that the \textquotedblleft quantization commutes with
reduction\textquotedblright\ map only makes sense if the quantization of
$T_{\mathbb{C}}^{\mathrm{rss}}/W$ is done in this way. (See Section
\ref{QvR.sec}.) We will show that for a suitable choice of the prequantum line
bundle, the desired outcome can be obtained by exploiting the freedom in the
standard procedure of geometric quantization with half-forms to choose the
prequantum line bundle.

\begin{proposition}
\label{weylNonzero.prop}The Weyl group acts freely on $T_{\mathbb{C}%
}^{\mathrm{rss}}$ and the analytically continued Weyl denominator
$\sigma_{\mathbb{C}}$ is nowhere vanishing on $T_{\mathbb{C}}^{\mathrm{rss}}.$
\end{proposition}

\begin{proof}
If $z\in T_{\mathbb{C}}^{\mathrm{rss}}$ were fixed by some nontrivial element
of $W:=N(T)/T,$ then $z$ would commute with some element of $N(T)$ not in $T,$
so that $z$ would not be regular semisimple. Meanwhile, from the formula for
$\sigma_{\mathbb{C}},$ we see that if $\sigma_{\mathbb{C}}(e^{H})=0,$ then
$\left\langle \alpha,H\right\rangle \in2\pi\mathbb{Z},$ from which it follows
that $\phi_{\alpha}(z)=1,$ showing that $z$ is not regular semisimple.
\end{proof}

We need to understand the canonical bundle $\mathcal{K}^{\prime}$ for the
quotient $T_{\mathbb{C}}^{\mathrm{rss}}/W.$ Since the volume form
$\beta^{\prime}$ on $T_{\mathbb{C}}$ is not invariant under the action of the
Weyl group, it does not descend to a form on the quotient. On the other hand,
since the Weyl denominator function $\sigma_{\mathbb{C}}$ is alternating, the
form%
\[
\sigma_{\mathbb{C}}\beta^{\prime}%
\]
is $W$-invariant and (by Proposition \ref{weylNonzero.prop}) nowhere vanishing
on $T_{\mathbb{C}}^{\mathrm{rss}}.$ Thus, we may regard this form a nowhere
vanishing form on the quotient. In particular, we have established that the
canonical bundle $\mathcal{K}^{\prime}$ for $T_{\mathbb{C}}^{\mathrm{rss}}/W$
is trivial. We may therefore take a trivial square root $\mathcal{K}%
_{1/2}^{\prime}$ of $\mathcal{K}^{\prime}$ with trivializing holomorphic
section $\sqrt{\sigma_{\mathbb{C}}\beta^{\prime}}.$

We would like to quantize $T_{\mathbb{C}}^{\mathrm{rss}}/W$ in such a way that
the sections have the form $Fe^{-\left\vert H\right\vert ^{2}/(2\hbar)}%
\otimes\sqrt{\beta^{\prime}},$ with $F$ being a Weyl-alternating holomorphic
function on $T_{\mathbb{C}}^{\mathrm{rss}}.$ We can formally rewrite such an
object as%
\[
\frac{Fe^{-\left\vert H\right\vert ^{2}/(2\hbar)}}{\sqrt{\sigma_{\mathbb{C}}}%
}\otimes\sqrt{\sigma_{\mathbb{C}}\beta^{\prime}},
\]
where as above we may regard $\sqrt{\sigma_{\mathbb{C}}\beta^{\prime}}$ as a
trivializing section of the canonical bundle of $T_{\mathbb{C}}^{\mathrm{rss}%
}/W.$ We now construct a line bundle $L^{\prime}$ over $T_{\mathbb{C}%
}^{\mathrm{rss}}/W$ in such way that the expression $Fe^{-\left\vert
H\right\vert ^{2}/(2\hbar)}/\sqrt{\sigma_{\mathbb{C}}}$ can be interpreted as
a holomorphic section of $L^{\prime}.$ (The prime distinguishes $L^{\prime}$
from the prequantum bundle $L$ over $T^{\ast}(K).$)

We define $L^{\prime}$ as the complex line bundle over $T_{\mathbb{C}%
}^{\mathrm{rss}}/W$ whose sections are $W$-alternating functions $f$ on
$T_{\mathbb{C}}^{\mathrm{rss}}.$ That is to say, the fiber of $L^{\prime}$
over each $W$-orbit $\mathcal{O}$ in $T_{\mathbb{C}}^{\mathrm{rss}}$ is the
one-dimensional complex vector space of $W$-alternating functions from
$\mathcal{O}$ into $\mathbb{C}.$ To each section $f$ of $L^{\prime}$ we
associate the formal object%
\begin{equation}
\frac{f}{\sqrt{\sigma_{\mathbb{C}}}}. \label{formalExpr}%
\end{equation}
We emphasize that $\sqrt{\sigma_{\mathbb{C}}}$ is not a single-valued function
on $T_{\mathbb{C}}^{\mathrm{rss}}$; the expression (\ref{formalExpr}) is
simply a mnemonic device that will help us remember the definition of the
Hermitian structure and connection on $L^{\prime}.$

Motivated by (\ref{formalExpr}) we define a Hermitian structure on $L^{\prime
}$ by setting
\begin{equation}
\left\vert f\right\vert _{L^{\prime}}(z)=\frac{\left\vert f(z)\right\vert
}{\left\vert \sigma_{\mathbb{C}}(z)\right\vert ^{1/2}},\quad z\in
T_{\mathbb{C}}^{\mathrm{rss}}, \label{TcHermitian}%
\end{equation}
for each $W$-alternating function $f.$ To define a connection on $L^{\prime},$
we observe that if $\sqrt{\sigma_{\mathbb{C}}}$ is any local square root of
$\sigma_{\mathbb{C}}$ and $X$ is a vector field, we have%
\[
X\left(  \frac{f}{\sqrt{\sigma_{\mathbb{C}}}}\right)  =\frac{1}{\sqrt
{\sigma_{\mathbb{C}}}}\left(  Xf-\frac{1}{2}\frac{X\sigma_{\mathbb{C}}}%
{\sigma_{\mathbb{C}}}f\right)  .
\]
Thus, the formal expression (\ref{formalExpr}) suggests to define a connection
on $L^{\prime}$ by setting
\begin{equation}
\nabla_{X}f=Xf-\frac{1}{2}\frac{X\sigma_{\mathbb{C}}}{\sigma_{\mathbb{C}}%
}f-\frac{i}{\hbar}\theta(X)f, \label{TcConnection}%
\end{equation}
whenever $X$ is a vector field on $T_{\mathbb{C}}^{\mathrm{rss}}/W,$ viewed as
a $W$-invariant vector field on $T_{\mathbb{C}}^{\mathrm{rss}},$ and $f$ is a
$W$-alternating function. Note that since $X$ is $W$-invariant and
$\sigma_{\mathbb{C}}$ is $W$-alternating, $X\sigma_{\mathbb{C}}/\sigma
_{\mathbb{C}}$ is $W$-invariant, so that $(X\sigma_{\mathbb{C}}/\sigma
_{\mathbb{C}})f$ is still $W$-alternating. Then, as usual in geometric
quantization, we define a smooth section $f$ of $L^{\prime}$ to be holomorphic
if
\[
\nabla_{X}f=0
\]
for all vectors of type $(0,1).$

\begin{proposition}
The curvature of $L^{\prime}$ with respect to the connection in
(\ref{TcConnection}) is $\omega/\hbar.$ A section $f$ is holomorphic if and
only if
\begin{equation}
\left(  X-\frac{i}{\hbar}\theta(X)\right)  f=0 \label{holoSectionTc}%
\end{equation}
for each vector field of type $(0,1),$ and this condition holds if and only if
$f$ has the form%
\begin{equation}
f=Fe^{-\left\vert H\right\vert ^{2}/(2\hbar)} \label{holoKahler}%
\end{equation}
for some Weyl-alternating holomorphic function $F$ on $T_{\mathbb{C}%
}^{\mathrm{rss}}.$
\end{proposition}

\begin{proof}
The connection (\ref{TcConnection}) differs from the usual one in
prequantization by the addition of the term involving $X\sigma_{\mathbb{C}%
}/\sigma_{\mathbb{C}}$. Locally, this change amounts to replacing $\theta$ by
$\theta^{\prime}=\theta+d\psi,$ where $\psi$ is a multiple of the locally
defined logarithm $\log(\sigma_{\mathbb{C}}).$ Since the curvature is computed
from $d\theta,$ this change does not affect the curvature. Similarly, since
$\sigma_{\mathbb{C}}$ is holomorphic, the term involving $\sigma_{\mathbb{C}}$
will vanish whenever $X$ is of type $(0,1),$ so the condition for a
holomorphic section is still (\ref{holoSectionTc}). Finally, since $T$ is also
a connected compact Lie group, the analysis we carried out in the quantization
of $T^{\ast}(K)$ applies also here, showing that solutions to
(\ref{holoSectionTc}) have the form (\ref{holoKahler}).
\end{proof}

We summarize the preceding discussion in the following definition.

\begin{definition}
[Quantization of the reduced phase space]\label{QuantReduced.def}Let
$L^{\prime}$ be the complex line bundle over $T_{\mathbb{C}}^{\mathrm{rss}}/W$
whose sections are $W$-alternating functions $f$ on $T_{\mathbb{C}%
}^{\mathrm{rss}},$ with Hermitian structure and connection on $L^{\prime}$ as
in (\ref{TcHermitian}) and (\ref{TcConnection}). Take a trivial square root
$\mathcal{K}_{1/2}^{\prime}$ of the canonical bundle $\mathcal{K}^{\prime}$
over $T_{\mathbb{C}}^{\mathrm{rss}}/W$ with trivializing section $\sqrt
{\sigma_{\mathbb{C}}\beta^{\prime}},$ with a Hermitian structure on
$\mathcal{K}_{1/2}^{\prime}$ defined similarly to (\ref{hermitianK12}). We
define our quantization of $T_{\mathbb{C}}^{\mathrm{rss}}/W$ as the space of
square-integrable holomorphic sections of $L^{\prime}\otimes\mathcal{K}%
_{1/2}^{\prime}.$ In accordance with the formal expression (\ref{formalExpr}),
we write elements $\psi$ of the quantum Hilbert space as%
\[
\psi=\frac{Fe^{-\left\vert H\right\vert ^{2}/(2\hbar)}}{\sqrt{\sigma
_{\mathbb{C}}}}\otimes\sqrt{\sigma_{\mathbb{C}}\beta^{\prime}},
\]
or, suggestively, as%
\[
\psi=Fe^{-\left\vert H\right\vert ^{2}/(2\hbar)}\otimes\sqrt{\beta^{\prime}},
\]
where $F$ is a $W$-alternating holomorphic function on~$T_{\mathbb{C}%
}^{\mathrm{rss}}.$
\end{definition}

The norm of such an element is computed as
\begin{equation}
\left\Vert \psi\right\Vert ^{2}=\frac{1}{\left\vert W\right\vert }%
\int_{T_{\mathbb{C}}^{\mathrm{rss}}}\left\vert F(z)\right\vert ^{2}%
e^{-\left\vert H\right\vert ^{2}/\hbar}\eta^{\prime}\varepsilon^{\prime},
\label{psiNormReduced}%
\end{equation}
where $\varepsilon^{\prime}$ is the Liouville volume measure on $T^{\ast
}(T)\cong T_{\mathbb{C}}$ and where $\eta^{\prime}$ is as in (\ref{etaPrime}).
In particular, we have identified%
\[
\mathrm{Quant}(K_{\mathbb{C}}/\!\!/\mathrm{Ad}_{K})=\mathrm{Quant}%
(T_{\mathbb{C}}^{\mathrm{rss}}/W)
\]
with the $W$-alternating subspace of $\mathrm{Quant}(T_{\mathbb{C}}),$ and
this identification is unitary up to a constant. (The constant arises because
of the factor of $1/\left\vert W\right\vert $ in (\ref{psiNormReduced}).)

\section{The \textquotedblleft quantization commutes with
reduction\textquotedblright\ map\label{QvR.sec}}

In this section, we construct a \textquotedblleft natural\textquotedblright%
\ map $B$ from the first-reduce-then-quantize Hilbert space $\mathrm{Quant}%
(K_{\mathbb{C}})^{\mathrm{Ad}_{K}}$ to the first-quantize-then-reduce Hilbert
space $\mathrm{Quant}(K_{\mathbb{C}}/\!\!/\mathrm{Ad}_{K})$. The map includes
a mechanism for converting half-forms of degree $n$ (over $K_{\mathbb{C}}$) to
half-forms of degree $r$ (over the regular part of the reduced phase). The
main result will be that $B$ coincides, after suitable identifications, with
the map of Florentino, Mour\~{a}o, and Nunes and therefore (Theorem
\ref{fmn.thm}) that $B$ is a constant multiple of a unitary map.

Recall that $L$ and $L^{\prime}$ denote the prequantum line bundles over
$T^{\ast}(K)\cong K_{\mathbb{C}}$ and the reduced phase space, respectively,
and that $\mathcal{K}_{1/2}$ and $\mathcal{K}_{1/2}^{\prime}$ denote chosen
square roots of the corresponding canonical bundles. We will introduce a
contraction mechanism that will allow us to convert invariant holomorphic
sections of $\mathcal{K}_{1/2}$ to holomorphic sections of $\mathcal{K}%
_{1/2}^{\prime}.$ A crucial factor of the Weyl denominator will arise in this
process. The process depends, however, on a certain choice of orientations and
it will \textit{not} be possible to make this choice consistently over all of
$\phi^{-1}(0)^{\mathrm{reg}}.$ Thus, the contraction procedure only makes
sense locally.

We also introduce a \textquotedblleft restriction\textquotedblright\ map for
mapping invariant holomorphic sections of $L$ to sections of $L^{\prime}.$
This map is similarly defined only locally. When we combine the two maps,
however, we get a globally defined map $B$ of invariant holomorphic sections
of $L\otimes\mathcal{K}_{1/2}$ to holomorphic sections of $L^{\prime}%
\otimes\mathcal{K}_{1/2}^{\prime}.$ The map $B$ is our \textquotedblleft
quantization commutes with reduction map\textquotedblright\ from
$\mathrm{Quant}(K_{\mathbb{C}})^{\mathrm{Ad}_{K}}$ to $\mathrm{Quant}%
(K_{\mathbb{C}}/\!\!/\mathrm{Ad}_{K}).$ Our main result is that $B$ after
suitable identifications, $B$ is the map described by the theorem of
Florentino, Mour\~{a}o, and Nunes and therefore that $B$ is a constant
multiple of a unitary map.

We now give a very brief summary of how $B$ is defined; details are given
below. To each element $\psi=Fe^{-\left\vert \xi\right\vert ^{2}/(2\hbar
)}\otimes\sqrt{\beta}$ of $\mathrm{Quant}(K_{\mathbb{C}})^{\mathrm{Ad}_{K}},$
we formally associate the quantity%
\[
\psi^{\prime}:=\left.  F\right\vert _{T_{\mathbb{C}}}e^{-\left\vert
H\right\vert ^{2}/(2\hbar)}\otimes(\sigma_{\mathbb{C}}\sqrt{\beta^{\prime}}),
\]
where, as we shall see, the factor of $\sigma_{\mathbb{C}}$ comes from the
contraction process. We then formally rewrite $\psi^{\prime}$ by moving the
factor of $\sigma_{\mathbb{C}}$ to the other side and multiplying and dividing
by $\sqrt{\sigma_{\mathbb{C}}}$ giving%
\[
\psi^{\prime}=\frac{(\sigma_{\mathbb{C}})(\left.  F\right\vert _{T_{\mathbb{C}%
}})e^{-\left\vert H\right\vert ^{2}/(2\hbar)}}{\sqrt{\sigma_{\mathbb{C}}}%
}\otimes\sqrt{\sigma_{\mathbb{C}}\beta^{\prime}}.
\]
We then note that $(\sigma_{\mathbb{C}})(\left.  F\right\vert _{T_{\mathbb{C}%
}})$ is a Weyl-alternating holomorphic function, so that $\psi^{\prime}$ is
indeed an element of $\mathrm{Quant}(K_{\mathbb{C}}/\!\!/\mathrm{Ad}_{K}%
)\cong\mathrm{Quant}(T_{\mathbb{C}})^{W_{-}}.$ Note that the function
$(\sigma_{\mathbb{C}})(\left.  F\right\vert _{T_{\mathbb{C}}})$ occurs also on
the right-hand side of Theorem \ref{fmn.thm}.

\subsection{Relating the canonical bundles\label{relatingCanonical.sec}}

We begin with considering the relationship between the canonical bundles over
$K_{\mathbb{C}}$ and over the reduced phase space $T_{\mathbb{C}%
}^{\mathrm{rss}}/W.$ Let $n$ be the complex dimension of $K_{\mathbb{C}}$ and
$r$ the complex dimension of $T_{\mathbb{C}}.$ Suppose $b$ is a holomorphic
$n$-form on $K_{\mathbb{C}}$ that is invariant under the adjoint action of
$K$. We hope to associate to $b$ a holomorphic $r$-form $\tilde{b}$ on the
regular part of the reduced phase space,
\[
\phi^{-1}(0)^{\mathrm{reg}}/\mathrm{Ad}_{K}\cong T_{\mathbb{C}}^{\mathrm{rss}%
}/W.
\]
The only reasonable way to do this is to restrict $b$ to $\phi^{-1}%
(0)^{\mathrm{reg}}$ and then contract with $n-r$ vector fields to convert $b$
from a $n$-form to an $r$-form. The only reasonable choice for the vector
fields are the vector fields $X^{\eta},$ $\eta\in\mathfrak{k},$ describing the
infinitesimal adjoint action of $K$ on $\phi^{-1}(0)^{\mathrm{reg}}.$

We now investigate this contraction process in detail. For each $(x,\xi)$ in
$\phi^{-1}(0)^{\mathrm{reg}},$ let $S_{(x,\xi)}$ be the stabilizer of
$(x,\xi)$---which is a maximal torus in $K$ because $(x,\xi)$ is assumed
regular---and let $\mathfrak{s}_{(x,\xi)}$ be the Lie algebra of $S_{(x,\xi
)}.$ Let $\eta_{1},\ldots,\eta_{n-r}$ be an orthonormal basis for the
orthogonal complement of $\mathfrak{s}_{(x,\xi)}^{\bot}$ of $\mathfrak{s}%
_{(x,\xi)}.$ We may then consider the contraction%
\begin{equation}
\tilde{b}:=i_{X^{\eta_{1}}\wedge\cdots\wedge X^{\eta_{n-r}}}(b).
\label{contractBeta}%
\end{equation}
This contraction is easily seen to be unchanged if we replace $\eta_{1}%
,\ldots,\eta_{n-r}$ by another orthonormal basis for $\mathfrak{s}_{(x,\xi
)}^{\bot}$ with the same orientation, but changes sign if we replace $\eta
_{1},\ldots,\eta_{n-r}$ by an orthonormal basis with the opposite orientation.

We now consider to the issue of trying to choose the orientations on
$\mathfrak{s}_{(x,\xi)}^{\bot}$ consistently over $\phi^{-1}(0)^{\mathrm{reg}%
}.$

\begin{proposition}
\label{orientation.prop}For each $(x,\xi)$ in $\phi^{-1}(0)^{\mathrm{reg}},$
let $\mathfrak{s}_{(x,\xi)}$ denote the Lie algebra of the stabilizer of
$(x,\xi)$ and $\mathfrak{s}_{(x,\xi)}^{\bot}$ denote the orthogonal complement
of $\mathfrak{s}_{(x,\xi)}$ in $\mathfrak{k}.$

\begin{enumerate}
\item For each $(x_{0},\xi_{0})\in\phi^{-1}(0)^{\mathrm{reg}},$ we can find an
open, $\mathrm{Ad}_{K}$-invariant set $U\subset\phi^{-1}(0)^{\mathrm{reg}}$
containing $(x_{0},\xi_{0})$ such that the orientation of $\mathfrak{s}%
_{(x,\xi)}^{\bot}$ can be chosen in a continuous, $\mathrm{Ad}_{K}$-invariant
fashion for all $(x,\xi)\in U.$

\item The orientation of $\mathfrak{s}_{(x,\xi)}^{\bot}$ \emph{cannot} be
chosen in a continuous, $\mathrm{Ad}_{K}$-invariant fashion for all
$(x,\xi)\in\phi^{-1}(0)^{\mathrm{reg}}.$
\end{enumerate}
\end{proposition}

Note that if $(x,\xi)\in\phi^{-1}(0)$ and $y\in K,$ then%
\[
S_{y\cdot(x,\xi)}=yS_{(x,\xi)}y^{-1},
\]
so that $\mathfrak{s}_{y\cdot(x,\xi)}=\mathrm{Ad}_{y}(\mathfrak{s}_{(x,\xi)})$
and $\mathfrak{s}_{y\cdot(x,\xi)}^{\bot}=\mathrm{Ad}_{y}(\mathfrak{s}%
_{(x,\xi)}^{\bot}).$ Thus, $\mathrm{Ad}_{K}$-invariance in the choice of
orientation would mean, explicitly, that $\mathrm{Ad}_{y},$ viewed as a map
from $\mathfrak{s}_{(x,\xi)}^{\bot}$ to $\mathfrak{s}_{y\cdot(x,\xi)}^{\bot}$,
is orientation preserving.

\begin{proof}
For Point 1, it is harmless to assume that $(x_{0,}\xi_{0})$ is equal to a
point $(t_{0},H_{0})$ in $T^{\ast}(T)^{\mathrm{reg}}.$ Since $W$ acts freely
on $T^{\ast}(T)^{\mathrm{reg}},$ we can take a neighborhood $V$ of
$(t_{0},H_{0})$ in $T^{\ast}(T)^{\mathrm{reg}}$ such that $V\cap w\cdot
V=\varnothing$ for all $w\neq1.$ Then $U:=\mathrm{Ad}_{K}\cdot V$ is an open
set in $\phi^{-1}(0)^{\mathrm{reg}}$ homeomorphic to $V\times(K/T).$ Each
stabilizer of $(t,H)\in V$ is simply $T,$ so we may choose orientations over
$V$ by using one fixed orientation on $\mathfrak{t}^{\bot}.$ Since the
stabilizer (i.e., $T$) of each point $(t,H)$ in $V$ is connected, the action
of the stabilizer on $\mathfrak{s}_{(t,H)}^{\bot}=\mathfrak{t}^{\bot}$ is
orientation preserving. This fact guarantees that we can extend the choice of
orientation from $V$ to $U$ in an unambiguous, invariant fashion.

For Point 2, note that we may make a continuous choice of orientation over
$T^{\ast}(T)^{\mathrm{reg}}$ simply by using one fixed orientation on
$\mathfrak{t}^{\bot}.$ Now, it is easily verified that $T^{\ast}%
(T)^{\mathrm{reg}}$ is connected; it follows that using one fixed orientation
on $\mathfrak{t}^{\bot}$ is the \textit{unique} continuous choice of
orientations over $T^{\ast}(T)^{\mathrm{reg}}.$

Now fix a Weyl group element $w$ with $\mathrm{sign}(w)=-1$ and pick a
representative $y$ of $w$ in $N(T).$ Then $\mathrm{Ad}_{y},$ viewed as a map
from $\mathfrak{t}$ to itself, is orientation reversing. But by the
connectedness of $K,$ $\mathrm{Ad}_{y},$ viewed as a map of $\mathfrak{k}$ to
itself, is orientation preserving. Thus, $\mathrm{Ad}_{y}$, viewed as a map of
$\mathfrak{t}^{\bot}$ to itself must be orientation reversing. Thus, any
continuous choice of orientation even over $T^{\ast}(T)^{\mathrm{reg}}$ fails
to be invariant under the adjoint action of $N(T)\subset K.$
\end{proof}

We now come to a key computation that ultimately explains the geometric origin
of the analytically continued Weyl denominator $\sigma_{\mathbb{C}}$ in the
\textquotedblleft quantization commutes with reduction\textquotedblright\ map.
To state our result, we now fix the normalization of the left-invariant
holomorphic forms $\beta$ and $\beta^{\prime}$ on $K_{\mathbb{C}}$ and
$T_{\mathbb{C}}.$ Let us fix an orientation of $\mathfrak{t}$ and an
orientation of $\mathfrak{k}$. This then determines an orientation of
$\mathfrak{t}^{\bot}$: If we take an oriented orthonormal basis $\eta
_{1},\ldots,\eta_{r}$ for $\mathfrak{t}$ and extend it to an oriented
orthonormal basis $\eta_{1},\ldots,\eta_{n}$ for $\mathfrak{k},$ then
$\eta_{r+1},\ldots,\eta_{n}$ should be an oriented basis for $\mathfrak{t}%
^{\bot}.$ Let us normalize $\beta$ and $\beta^{\prime}$ so that at the
identity, we have
\begin{equation}
\beta^{\prime}(\eta_{1},\ldots,\eta_{r})=1;\quad\beta(\zeta_{1},\ldots
,\zeta_{n})=1 \label{betaNormalizations}%
\end{equation}
whenever $\eta_{1},\ldots,\eta_{r}$ and $\zeta_{1},\ldots,\zeta_{n}$ are
oriented orthonormal bases for $\mathfrak{t}$ and $\mathfrak{k},$
respectively. Note that $\beta$ and $\beta^{\prime}$ are defined on the
\textit{complex} vector spaces $\mathfrak{k}_{\mathbb{C}}$ and $\mathfrak{t}%
_{\mathbb{C}},$ respectively, but that the normalizations are fixed on bases
of the underlying \textit{real} vector spaces $\mathfrak{k}$ and
$\mathfrak{t}.$

\begin{proposition}
Consider a point $(x,\xi)\in\phi^{-1}(0)^{\mathrm{reg}}$. Let $U\subset
\phi^{-1}(0)^{\mathrm{reg}}$ be an open, $\mathrm{Ad}_{K}$-invariant set
containing $(x,\xi)$ as in Proposition \ref{orientation.prop}, and let us fix
a continuous, $\mathrm{Ad}_{K}$-invariant choice of orientation on
$\mathfrak{s}_{(x,\xi)}^{\bot},$ $(x,\xi)\in U.$ Let $b$ be a holomorphic,
$\mathrm{Ad}_{K}$-invariant $n$-form on $T^{\ast}(K)\cong K_{\mathbb{C}}$ and
let $\tilde{b}$ be the $r$-form on $U$ defined by (\ref{contractBeta}). Then
$\tilde{b}$ descends to a holomorphic $r$-form on $U^{\prime},$ the image of
$U$ in $\phi^{-1}(0)^{\mathrm{reg}}/\mathrm{Ad}_{K}\cong T_{\mathbb{C}%
}^{\mathrm{rss}}/W$.

Suppose, specifically, that $b=\beta$. Let $[z_{0}]$ be a point in $U^{\prime
}\subset T_{\mathbb{C}}^{\mathrm{rss}}/W,$ and let $z_{0}$ be a representative
of $z$ in $T_{\mathbb{C}}^{\mathrm{rss}}.$ If we then identify a neighborhood
of $[z_{0}]$ in $T_{\mathbb{C}}^{\mathrm{rss}}/W$ with a neighborhood of
$z_{0}$ in $T_{\mathbb{C}}^{\mathrm{rss}},$ we have%
\begin{equation}
\tilde{\beta}(z)=\pm\sigma_{\mathbb{C}}(z)^{2}\beta^{\prime}(z),\quad z\in
T_{\mathbb{C}}, \label{betaTilde}%
\end{equation}
where the sign depends both on the choice of orientations on $U$ and on the
choice of the representative $z$ of $[z].$
\end{proposition}

Note that since the function $\sigma_{\mathbb{C}}^{2}$ on $T_{\mathbb{C}}$ is
Weyl invariant and the form $\beta^{\prime}$ is Weyl alternating, the form
$\sigma_{\mathbb{C}}^{2}\beta^{\prime}$ is Weyl alternating. We thus see very
clearly the effect of the nonexistence of a global choice of orientation over
$\phi^{-1}(0)^{\mathrm{reg}}$: The form $\tilde{\beta}$ in (\ref{betaTilde})
is not a Weyl-invariant form on $T_{\mathbb{C}}^{\mathrm{rss}}$ and it
therefore does not descend to a form on $T_{\mathbb{C}}^{\mathrm{rss}}/W.$ On
the other hand, if $z_{0}\in T_{\mathbb{C}}^{\mathrm{rss}},$ we can pick a
small neighborhood $V$ of $z_{0}$ such that the sets $w\cdot V,$ $w\in W,$ are
disjoint. Then there \textit{does} exist a Weyl-invariant form on the union of
the $w\cdot V$'s whose restriction to $V$ is $\sigma_{\mathbb{C}}(z)^{2}%
\beta^{\prime}(z),$ namely the one whose restriction to $w\cdot V$ is
$\mathrm{sign}(w)\sigma_{\mathbb{C}}(z)^{2}\beta^{\prime}(z).$

\begin{proof}
Let $b$ be as in the first part of the proposition. Define a form $\tilde{b}$
on $U\subset\phi^{-1}(0)^{\mathrm{reg}}$ by (\ref{contractBeta}). We think of
$\tilde{b}$ as a locally defined form on $\phi^{-1}(0)^{\mathrm{reg}},$
meaning that we only plug into $\tilde{b}$ vectors that are tangent to
$\phi^{-1}(0)^{\mathrm{reg}}.$ Since $b$ is assumed to be invariant under the
adjoint action of $K$ and since the orientations over $U$ are chosen
invariantly, it is easy to check that $\tilde{b}$ is invariant under the
adjoint action of $K$ on $U.$

Fix some $(x,\xi)$ in $\phi^{-1}(0)^{\mathrm{reg}},$ let $V$ denote the
tangent space to $\phi^{-1}(0)^{\mathrm{reg}}$ at $(x,\xi)$ and let $W\subset
V$ denote the tangent space to the $\mathrm{Ad}_{K}$-orbit through $(x,\xi).$
Then $\tilde{b}(Y_{1},\ldots,Y_{r})=0$ if even one of the $Y_{j}$'s is in $W.$
(After all, $\tilde{b}$ is obtained from $b$ by contracting with a basis
$X^{\eta_{1}},\ldots,X^{\eta_{n-r}}$ for $W$.) Thus, $\tilde{b}$ descends to a
$r$-linear, alternating form on $V/W,$ which is just the tangent space to the
reduced phase space. Since $\tilde{b}$ is invariant under adjoint action of
$K,$ the value of $\tilde{b}$ at a point in the reduced space is independent
of the choice of point in the corresponding $K$-orbit in $\phi^{-1}%
(0)^{\mathrm{reg}}.$

It is presumably possible to verify that $\tilde{b}$ is holomorphic by an
argument similar to the one in \cite[Section 3.2]{HallKirwin}. In this
situation, however, we can work by direct computation. We first note that if
$F$ is an $\mathrm{Ad}_{K}$-invariant holomorphic function on $K_{\mathbb{C}%
},$ then the restriction of $F$ to $\phi^{-1}(0)^{\mathrm{reg}}$ is also
$\mathrm{Ad}_{K}$-invariant, so that this restriction descends to a function
$\tilde{F}$ on $\phi^{-1}(0)^{\mathrm{reg}}/\mathrm{Ad}_{K}.$ It should be
clear from the way the complex structure on $\phi^{-1}(0)^{\mathrm{reg}%
}/\mathrm{Ad}_{K}$ is defined that $\tilde{F}$ is again holomorphic.
(Explicitly, if we identify $\phi^{-1}(0)^{\mathrm{reg}}/\mathrm{Ad}_{K}$ with
$T_{\mathbb{C}}^{\mathrm{rss}}/W,$ the $\tilde{F}$ is simply the function on
$T_{\mathbb{C}}^{\mathrm{rss}}/W$ obtained from the restriction of $F$ to
$T_{\mathbb{C}}^{\mathrm{rss}},$ which is again holomorphic.) Now, every
$\mathrm{Ad}_{K}$-invariant holomorphic $n$-form on $K_{\mathbb{C}}$ is
expressible as an $\mathrm{Ad}_{K}$-invariant holomorphic function times the
form $\beta.$ It thus suffices to check holomorphicity in the case $b=\beta,$
which we do in what follows.

Under our standing identification (\ref{PhiDef}) of $T^{\ast}(K)$ with
$K_{\mathbb{C}},$ the adjoint action of $K$ on $T^{\ast}(K)$ corresponds to
the adjoint action of $K$ on $K_{\mathbb{C}}.$ We identify the tangent space
at each point in $K_{\mathbb{C}}$ with $\mathfrak{k}_{\mathbb{C}}$ by means of
left translation. With this identification, the value of the vector field
$X^{\eta}$ at a point $g\in K_{\mathbb{C}}$ is easily computed to be
\[
X^{\eta}=\mathrm{Ad}_{g^{-1}}(\eta)-\eta.
\]
Suppose now that $z\in T_{\mathbb{C}}^{\mathrm{rss}}$ and that $\eta
\in\mathfrak{t}^{\bot}.$ Then%
\begin{equation}
X^{\eta}=\mathrm{Ad}_{z^{-1}}(\eta)-\eta\label{XetaFormula}%
\end{equation}
is easily seen to lie in $\mathfrak{t}_{\mathbb{C}}^{\bot}.$

Now let $\eta_{1},\ldots,\eta_{n-r}$ be an oriented orthonormal basis for
$\mathfrak{s}_{z}^{\bot}=\mathfrak{t}^{\bot}.$ Then%
\begin{align}
&  \beta(X^{\eta_{1}},\ldots,X^{\eta_{n-r}},Y_{1},\ldots,Y_{r})\nonumber\\
&  =\det(A)\beta(\eta_{1},\ldots,\eta_{n-r},Y_{1},\ldots,Y_{r})\nonumber\\
&  =\det(A)\beta(Y_{1},\ldots,Y_{r},\eta_{1},\ldots,\eta_{n-r})
\label{betaContract}%
\end{align}
for all $Y_{1},\ldots,Y_{r}\in\mathfrak{k}_{\mathbb{C}},$ where
$A:\mathfrak{t}_{\mathbb{C}}^{\bot}\rightarrow\mathfrak{t}_{\mathbb{C}}^{\bot
}$ is the unique linear transformation such that $A(\eta_{j})=X^{\eta_{j}}.$
(There is no minus sign in the second equality because $n-r$ is the number of
roots, which is even.) From (\ref{XetaFormula}), we can identify $A$ as%
\[
A=\mathrm{Ad}_{z^{-1}}^{\prime}-I,
\]
where $\mathrm{Ad}_{z^{-1}}^{\prime}$ is the restriction of $\mathrm{Ad}%
_{z^{-1}}$ to $\mathfrak{t}_{\mathbb{C}}^{\bot}.$

Now, the eigenvalues of $\mathrm{Ad}_{z^{-1}}^{\prime}$ are the numbers of the
form $\phi_{\alpha}(z^{-1}),$ $\alpha\in R,$ where $\phi_{\alpha}$ is the root
homomorphism in (\ref{phiAlphaDef}). Thus,%
\[
\det(A)=\prod_{\alpha\in R}(\phi_{\alpha}(z^{-1})-1).
\]
After group the roots into pairs, $\{\alpha,-\alpha\}$ with $\alpha\in R^{+},$
this result simplifies to
\begin{align*}
\det(A)  &  =\prod_{\alpha\in R^{+}}[(e^{i\left\langle \alpha,H\right\rangle
/2}-e^{-i\left\langle \alpha,H\right\rangle /2})(e^{-i\left\langle
\alpha,H\right\rangle /2}-e^{i\left\langle \alpha,H\right\rangle /2})]\\
&  =(-1)^{m}\sigma_{\mathbb{C}}(z)^{2},
\end{align*}
where $m$ is the number of positive roots. In particular, if $Y_{1}%
,\ldots,Y_{r}$ are in $\mathfrak{t}_{\mathbb{C}},$ we have
\[
\beta(X^{\eta_{1}},\ldots,X^{\eta_{n-r}},Y_{1},\ldots,Y_{r})=(-1)^{m}%
\sigma_{\mathbb{C}}(z)^{2}\beta^{\prime}(Y_{1},\ldots,Y_{r}).
\]
This equality certainly holds up to a constant because both sides are
$r$-linear alternating functions of $Y_{1},\ldots,Y_{r}.$ The constant may
then be checked from (\ref{betaContract}) when $Y_{1},\ldots,Y_{r}$ form an
oriented orthonormal basis for $\mathfrak{t}.$
\end{proof}

\subsection{Relating the half-form bundles\label{relatingHalf.sec}}

We now observe that the locally defined contraction process on sections of the
canonical bundle $\mathcal{K}$ extends to a locally defined contraction
process on sections of $\mathcal{K}_{1/2}.$ The idea is simple. We start with
an invariant holomorphic section $c$ of $\mathcal{K}_{1/2}$ and square it to
an invariant holomorphic section $b:=c\otimes c$ of $\mathcal{K}.$ Then we
contract $b$ to a locally defined section $b^{\prime}$ of $\mathcal{K}%
^{\prime}.$ Finally, we look for a locally defined holomorphic section
$c^{\prime}$ of $\mathcal{K}_{1/2}^{\prime}$ with $c^{\prime}\otimes
c^{\prime}=b^{\prime}.$ It is easy to see that the preceding procedure can be
carried out locally, which is all that we hope for at the moment. (In the
setting of \cite{HallKirwin}, this contraction process on the half-form
bundles can be done globally; see Theorem 3.1 there.)

Using the computations in the previous subsection, we can read off the results
of the contraction process on the half-form bundles, as follows.

\begin{proposition}
The contraction process on the canonical bundles induces a locally defined
contraction process on the half-form bundles. For each $[z_{0}]\in
T_{\mathbb{C}}^{\mathrm{rss}}/W,$ we choose a representative $z_{0}$ of
$[z_{0}]$ in $T_{\mathbb{C}}^{\mathrm{rss}},$ which we use to identify
$T_{\mathbb{C}}^{\mathrm{rss}}/W$ locally with $T_{\mathbb{C}}^{\mathrm{rss}%
}.$ Then the contraction process on half-forms is determined by its action on
$\sqrt{\beta},$ which is given by
\begin{equation}
\sqrt{\beta}\mapsto\pm\sigma_{\mathbb{C}}\sqrt{\beta^{\prime}}.
\label{contractSqrt}%
\end{equation}

\end{proposition}

At the moment, the sign in (\ref{contractSqrt}) is undefined, since the
section $c^{\prime}$ in the above description is only unique up to a sign.

\subsection{Relating the prequantum bundles}

We now construct a natural \textit{local} map from the space of invariant
holomorphic sections of the prequantum bundle $L$ over $T^{\ast}(K)\cong
K_{\mathbb{C}}$ to the space of holomorphic sections of the prequantum bundle
$L^{\prime}$ over $T^{\ast}(T)^{\mathrm{reg}}/W\cong T_{\mathbb{C}%
}^{\mathrm{rss}}/W.$ Suppose $f$ is an invariant section of $L,$ that is, that
$Q_{\mathrm{pre}}(\phi_{\eta})f=0$ for all $\eta\in\mathfrak{k}.$ Then by
(\ref{QpreTheta}), $f$ is a function that is invariant under the adjoint
action of $K.$ Invariant \textit{holomorphic} sections of $L$ then have the
form $Fe^{-\left\vert \xi\right\vert ^{2}/(2\hbar)},$ where $F$ is an
$\mathrm{Ad}_{K}$-invariant holomorphic function.

Pick a point $[z_{0}]$ in $T_{\mathbb{C}}^{\mathrm{rss}}/W$ and a
representative $z_{0}$ of $[z_{0}]$ in $T_{\mathbb{C}}^{\mathrm{rss}}.$ Let us
identify a small neighborhood of $[z_{0}]\in T_{\mathbb{C}}^{\mathrm{rss}}/W$
with a neighborhood $V$ of $z\in T_{\mathbb{C}}^{\mathrm{rss}},$ where $V$ is
chosen to be simply connected and so that the sets $w\cdot V,~w\in W,$ are
disjoint. Pick a local holomorphic square root of $\sigma_{\mathbb{C}}$ on
$V$, denoted $[\sigma_{\mathbb{C}}]^{1/2}.$ Then to each invariant holomorphic
section $f=Fe^{-\left\vert \xi\right\vert ^{2}/(2\hbar)}$ of $L,$ we associate
the function
\begin{equation}
f^{\prime}(z)=[\sigma_{\mathbb{C}}(z)]^{1/2}F(z)e^{-\left\vert H\right\vert
^{2}/(2\hbar)},\quad z\in V, \label{fprime}%
\end{equation}
which we also write as the formal object%
\[
\frac{\lbrack\sigma_{\mathbb{C}}(z)]^{1/2}F(z)e^{-\left\vert H\right\vert
^{2}/(2\hbar)}}{\sqrt{\sigma_{\mathbb{C}}}}.
\]
Note that the $\sqrt{\sigma_{\mathbb{C}}}$ in the denominator is just a formal
expression that reminds of the definition of the Hermitian structure and
connection on $L^{\prime}.$ The factor of $[\sigma_{\mathbb{C}}(z)]^{1/2},$ by
contrast, is an actual holomorphic function on $V\subset T_{\mathbb{C}%
}^{\mathrm{rss}}.$ Now the function $f^{\prime}$ on $V$ extends to a
$W$-alternating function on $W\cdot V,$ so that $f^{\prime}$ may be thought of
as a locally defined holomorphic section of $L^{\prime}$. (Recall that
sections of $L^{\prime}$ are by definition $W$-alternating functions on
$T_{\mathbb{C}}^{\mathrm{rss}}.$)

The motivation for the factor of $[\sigma_{\mathbb{C}}(z)]^{1/2}$ in
(\ref{fprime}) is this: The correspondence $f\mapsto f^{\prime}$ preserves the
pointwise magnitude of sections, in the sense that if $[z]\in T_{\mathbb{C}%
}^{\mathrm{rss}}/W$ comes from the $\mathrm{Ad}_{K}$-orbit of $(x,\xi)\in
\phi^{-1}(0)^{\mathrm{reg}},$ then
\[
\left\vert f^{\prime}(z)\right\vert _{L^{\prime}}=|f(x,\xi)|.
\]
(Recall from (\ref{TcHermitian}) that the pointwise magnitude of a section of
$L$ is simply the absolute value of the corresponding function.)

As in the case of the map between sections of the half-form bundle, the
correspondence $f\mapsto f^{\prime}$ \textit{does not} extend to a globally
defined map of invariant holomorphic sections of $L$ to holomorphic sections
of $L^{\prime},$ because $[\sigma_{\mathbb{C}}]^{1/2}$ does not extend to a
globally defined $W$-alternating holomorphic function on $T_{\mathbb{C}%
}^{\mathrm{rss}}.$

\subsection{The map\label{theMap.sec}}

Recall from Section \ref{reducedQuant.sec} that elements of $\mathrm{Quant}%
(K_{\mathbb{C}}/\!\!/\mathrm{Ad}_{K})$ are holomorphic sections of the bundle
$L^{\prime}\otimes\mathcal{K}_{1/2}^{\prime}$ over $T_{\mathbb{C}%
}^{\mathrm{rss}}/W.$ We write such a section in the form%
\[
\frac{Fe^{-\left\vert H\right\vert ^{2}/(2\hbar)}}{\sqrt{\sigma_{\mathbb{C}}}%
}\otimes\sqrt{\sigma_{\mathbb{C}}\beta^{\prime}},
\]
or in the suggestive alternative form%
\[
Fe^{-\left\vert H\right\vert ^{2}/(2\hbar)}\otimes\sqrt{\beta^{\prime}},
\]
where $F$ is a $W$-alternating holomorphic function on $T_{\mathbb{C}%
}^{\mathrm{rss}}.$ We now show that the \textit{local} mappings described in
the two previous sections combine into a \textit{global} mapping $B$ of
invariant holomorphic sections of $L\otimes\mathcal{K}_{1/2}$ to holomorphic
sections $L^{\prime}\otimes\mathcal{K}_{1/2}^{\prime}.$ After suitable
identifications, $B$ will coincide with the correspondence in the theorem of
Florentino, Mour\~{a}o, and Nunes (Theorem \ref{fmn.thm}), showing that $B$ is
a multiple of a unitary map.

If $\psi$ is an invariant section of $L\otimes\mathcal{K}_{1/2},$ we write%
\[
\psi=Fe^{-\left\vert \xi\right\vert ^{2}/(2\hbar)}\otimes\sqrt{\beta}.
\]
If we apply the local mappings of the two previous subsections to
$Fe^{-\left\vert \xi\right\vert ^{2}/(2\hbar)}$ and to $\sqrt{\beta},$ we
obtain the local section%
\[
\psi^{\prime}=\frac{[\sigma_{\mathbb{C}}(z)]^{1/2}F(z)e^{-\left\vert
H\right\vert ^{2}/(2\hbar)}}{\sqrt{\sigma_{\mathbb{C}}}}\otimes(\sigma
_{\mathbb{C}}\sqrt{\beta^{\prime}}).
\]
Moving a factor of $[\sigma_{\mathbb{C}}]^{1/2}$ from right to left in the
tensor product allows us to rewrite this as%
\[
\psi^{\prime}=\frac{\sigma_{\mathbb{C}}(z)F(z)e^{-\left\vert H\right\vert
^{2}/(2\hbar)}}{\sqrt{\sigma_{\mathbb{C}}}}\otimes\sqrt{\sigma_{\mathbb{C}%
}\beta^{\prime}}%
\]
or, suggestively, as
\[
\psi^{\prime}=(\sigma_{\mathbb{C}}(z)F(z))e^{-\left\vert H\right\vert
^{2}/(2\hbar)}\otimes\sqrt{\beta^{\prime}}.
\]

We now observe that $\psi^{\prime}$ is actually a globally defined holomorphic
section of $L^{\prime}\otimes\mathcal{K}_{1/2}^{\prime}.$ After all, if $F$ is
a holomorphic class function on $K_{\mathbb{C}},$ then the restriction of $F$
to $T_{\mathbb{C}}^{\mathrm{rss}}$ is Weyl invariant, so that the function%
\[
(\sigma_{\mathbb{C}})(\left.  F\right\vert _{T_{\mathbb{C}}^{\mathrm{rss}}})
\]
is Weyl alternating.

\begin{theorem}
[Quantization commutes with reduction]The locally defined maps in the previous
subsections combine to give a globally defined map from the space of invariant
holomorphic sections of $L\otimes\mathcal{K}_{1/2}$ to the space of
holomorphic sections of $L^{\prime}\otimes\mathcal{K}_{1/2}^{\prime}.$ Thus,
we obtain a geometrically natural map
\[
B:\mathrm{Quant}(K_{\mathbb{C}})^{\mathrm{Ad}_{K}}\rightarrow\mathrm{Quant}%
(K_{\mathbb{C}}/\!\!/\mathrm{Ad}_{K})\cong\mathrm{Quant}(T_{\mathbb{C}%
})^{W_{-}},
\]
which may be computed explicitly as follows:%
\[
B(Fe^{-\left\vert \xi\right\vert ^{2}/(2\hbar)}\otimes\sqrt{\beta}%
)=\frac{(\sigma_{\mathbb{C}}F)e^{-\left\vert H\right\vert ^{2}/(2\hbar)}%
}{\sqrt{\sigma_{\mathbb{C}}}}\otimes\sqrt{\sigma_{\mathbb{C}}\beta^{\prime}}.
\]
The map $B$ is a constant multiple of a unitary map.
\end{theorem}

\begin{proof}
We have already established that $B$ is well defined and computed as in the
theorem. As a map of $\mathrm{Quant}(K_{\mathbb{C}})^{\mathrm{Ad}_{K}}$ to
$\mathrm{Quant}(T_{\mathbb{C}})^{W_{-}},$ $B$ is the map sending the
holomorphic class function $F$ on $K_{\mathbb{C}}$ to the function
$(\sigma_{\mathbb{C}})(\left.  F\right\vert _{T_{\mathbb{C}}})$ on
$T_{\mathbb{C}}.$ The unitarity claim then follows from Proposition
\ref{geoHeat.prop} and Theorem \ref{fmn.thm}.
\end{proof}

\section{Funding}

This work was supported by the National Science Foundation [DMS-1301534 to B.C.H.].

\section{Acknowledgments}

We thank Sam Evens for numerous helpful suggestions and references. We also
thank the referee for reading the manuscript carefully and making several corrections.

\end{document}